\documentclass[11pt]{article}

\usepackage{amsmath,amsthm,nicefrac}
\usepackage{amsfonts, amstext}

\usepackage[compact]{titlesec}
\usepackage{fullpage}
\usepackage[font={small,it}]{caption}

\usepackage{setspace}
\usepackage{comment}

\usepackage{enumitem}
\usepackage[breaklinks]{hyperref}
\hypersetup{colorlinks=true,%
            citebordercolor={.6 .6 .6},linkbordercolor={.6 .6 .6},%
citecolor=blue,urlcolor=black,linkcolor=blue,pagecolor=black}

\usepackage[nameinlink]{cleveref}
\Crefname{algocf}{Algorithm}{Algorithms}
\crefname{algocfline}{line}{lines}
\Crefname{invariant}{Invariant}{Invariants}
\Crefname{claim}{Claim}{Claims}
\Crefname{subclaim}{Subclaim}{Subclaims}

\usepackage{epsfig}
\usepackage{amsthm,amssymb}
\usepackage{amsthm,amssymb}

\usepackage{mathrsfs}
\usepackage{xspace}
\usepackage{soul}
\usepackage{latexsym}
\usepackage{bbm}

\usepackage{framed}

\usepackage[dvipsnames]{xcolor}
\definecolor{DarkGray}{rgb}{0.66, 0.66, 0.66}
\definecolor{DarkPowderBlue}{rgb}{0.0, 0.2, 0.6}
\definecolor{fluorescentyellow}{rgb}{0.8, 1.0, 0.0}

\usepackage[ruled,vlined,linesnumbered,algonl]{algorithm2e}
\SetEndCharOfAlgoLine{}
\SetKwComment{Comment}{\footnotesize$\triangleright$\ }{}

\SetCommentSty{mycommfont}

\usepackage{thmtools,thm-restate}

\allowdisplaybreaks

\newcounter{note}[section]
\renewcommand{\thenote}{\thesection.\arabic{note}}

\newcommand{\rznote}[1]{\refstepcounter{note}$\ll${\bf Rudy~\thenote:}
	{\sf \color{red} #1}$\gg$\marginpar{\tiny\bf RZ~\thenote}}

\sethlcolor{fluorescentyellow}

\newcommand{\initOneLiners}{%
    \setlength{\itemsep}{0pt}
    \setlength{\parsep }{0pt}
    \setlength{\topsep }{0pt}
}

\newtheorem{theorem}{Theorem}[section]
\newtheorem{prop}[theorem]{Proposition}
\newtheorem{lem}[theorem]{Lemma}
\newtheorem{cor}[theorem]{Corollary}

\theoremstyle{definition}
\newtheorem{definition}[theorem]{Definition}
\newtheorem{example}[theorem]{Example}

\theoremstyle{remark}

\newcommand{\EE}{\mathbb{E}}
\newcommand{\RR}{\mathbb{R}}
\newcommand{\NN}{\mathbb{N}}

\usepackage{bbm}

\newcommand{\alg}{\textsc{Alg}}
\newcommand{\opt}{\textsc{Opt}}
\newcommand{\ind}{\mathbbm{1}}
\newcommand{\prob}{\mathbb{P}}
\newcommand{\gain}{\textsc{Gain}\xspace}

\newcommand{\ee}{\textsc{EE}\xspace}
\newcommand{\optee}{\textsc{Opt}_{\ee}\xspace}
\newcommand{\opteec}{\textsc{Opt}_{\ee}^C\xspace}
\newcommand{\pol}{\mathcal{P}}

\newcommand{\bayes}{\textsc{BayesProbe}\xspace}
\newcommand{\probe}{\textsc{StochProbe}\xspace}
\newcommand{\staralg}{\textsc{EEStar}\xspace}

\newcommand{\genalg}{\textsc{AdaptGeneral}\xspace}
\newcommand{\decompalg}{\textsc{StarDecomp}\xspace}
\newcommand{\knapsack}{\textsc{NonAdapt}\xspace}
\newcommand{\lp}{\textsc{LP}\xspace}
\newcommand{\fracgreedy}{\textsc{FracGreedy}\xspace}
\newcommand{\adaptiveknapsack}{\textsc{AdaptSubmodKnapsack}\xspace}

\title{Bayesian Probing on Graphs}
\author{Anupam Gupta \\ New York University \and  Benjamin Moseley \\ Carnegie Mellon University \and  Rudy Zhou \\ Microsoft}

\begin{document}
	
	\maketitle

    \begin{abstract} 


We introduce a stochastic probing problem with correlated items. In
our model, which we call \emph{Bayesian Probing}, the correlations are
modeled by an underlying graph $G$. Each vertex is independently
active with a known probability. Each item corresponds to an edge in
the graph. Probing an edge has some cost, gives some reward if both
endpoints are active, and also reveals the state of its
endpoints. Hence a probe induces a Bayesian update on the remaining edges. The goal is to adaptively probe items/edges subject to a knapsack constraint to maximize the expected total reward obtained from the probed edges.

Bayesian Probing generalizes stochastic knapsack and stochastic
probing by allowing correlations between items. Moreover, it gives a
tractable model for the \emph{Bayesian Active Search} problem, a
popular problem considered in the machine learning community. In
Bayesian Active Search, the goal is to find items in a particular
class by adaptively probing at most, say $k$, items. Given a prior
distribution over items, we want to compute a Bayesian policy to
maximize the number of such items found. For this general problem with
arbitrary priors, there are strong lower bounds on efficiently
computing good policies.

In this paper, we design efficient approximation algorithms for
Bayesian
Probing: 
\begin{itemize}
\item We give asymptotically tight bounds on the \emph{adaptivity
    gap}: we design a non-adaptive $O(d_{\max} \cdot S)$-approximation
  algorithm, where $d_{\max}$ is the maximum vertex degree of $G$ and
  $S$ is the ratio of the maximum and minimum edge costs. We also show
  a matching $\Omega(d_{\max} \cdot S)$ lower bound for non-adaptive policies.
\item We design an adaptive policy that achieves a
  $O(dens(G) \cdot S^2)$-approximation, where $dens(G)$ is the density
  of the densest subgraph of $G$.
\end{itemize}
These results give the first efficient approximation algorithms for
Bayesian Active Search, for a class of practically-relevant prior
distributions. 


    \end{abstract}
    
	\thispagestyle{empty}
	
	\newpage
	\setcounter{page}{1}
	
	\section{Introduction}\label{sec-intro}


Active search is a fundamental problem in applications like fraud detection and drug discovery \cite{DBLP:journals/jcamd/GarnettGVB15}, where the objective is to identify elements of a target class within a large search space under a limited budget. In this paradigm, items within the search space -- such as financial transactions or chemical compounds -- are often interrelated; transactions may share metadata, or chemical compounds may have common substructures. Consequently, evaluating a single item can provide information about others, a process often modeled via a Bayesian update on the likelihood of other items belonging to the target class. Bayesian optimization provides a popular framework for sequentially selecting items to evaluate, leveraging the outcomes of past evaluations to inform future choices. See the book by Garnett \cite{garnett_bayesoptbook_2023} and the work of Jiang et al. \cite{JiangMAMG18,JiangMCSMG17} for an overview.

More formally, in the \emph{Active Search problem} \cite{GarnettKXSM12}, as input we are given a finite set of elements $X$ and budget $k \in \mathbb{N}$. We assume there is an \emph{unknown} subset $T \subset X$ of \emph{targets} but with known distribution $T \sim \mathcal{D}$. Note that in general, there may be correlations among different elements being in the target set. For each $x \in X$, we do not know if $x \in T$ initially, but we can request a binary test for $x$, which returns either ``$x \in T$'' or ``$x \notin T$''. The goal is to adaptively choose a set of at most $k$ tests, say $P \subset X$ to maximize the expected total number of targets found: $\EE[\sum_{x \in P} \ind(x \in T)]$.

Prior work has focused on computing the optimal policy for a given objective function -- typically maximizing the number of items found in the class as stated above -- but unfortunately the optimal policy requires exponential space and time in general. Further, for general correlation models, no polynomial time algorithm can have better than a polylogartihmic approximation ratio for this objective \cite{JiangMCSMG17}.  While this lower bound is known, there has been no progress on upper bounds on the approximation ratio. Instead, researchers have focused on heuristic methods with no performance guarantees \cite{DBLP:conf/kdd/WangGS13,JiangMCSMG17}.

The main motivation of this paper is to give a \emph{concrete correlation model} for Bayesian active search that enables us to efficiently compute \emph{provably near-optimal} policies. This work gives the foundation for algorithms with provable guarantees for Bayesian active search problem by leveraging techniques from stochastic combinatorial optimization.

\paragraph{Bayesian Probing Problem} Motivated by Bayesian active search, we introduce the \emph{Bayesian probing problem} (\bayes): We are given as input an undirected graph $G = (V,E)$\footnote{We allow parallel edges, but not self loops; given a graph with self loops, for each such edge we create a dummy vertex with probability $1$ and make one endpoint of the self loop this dummy vertex.}, where each vertex $v \in V$ is independently \emph{active} with probability $p_v \in [0,1]$. The probabilities are \emph{known} to the algorithm apriori. We assume that the probabilities $p_v$'s are known, but the active vertices are not. We let $X_v$ be the indicator random variable for vertex $v$ being active. Each edge $uv \in E$ has a random \emph{reward} $R_{uv} = w_{uv} \cdot X_u X_v$, where $w_{uv} \geq 0$ and $s_{uv} > 0$ are the \emph{weight} and \emph{size} of this edge, respectively.

We can access the graph only through \emph{edge probes}. We are given a budget $B > 0$. Probing edge $e$ uses $s_e$ of the budget and gives us reward $R_e$. Further, we learn the outcome of both endpoint vertices, $X_u$ and $X_v$. We allow probes to be chosen adaptively, so the subsequent probes can depend on the information learned from prior ones. The goal is to adaptively probe edges of total size at most $B$ to maximize the expected reward obtained from the probed edges -- precisely:
\[\EE \big[ \sum_{e \in P} R_e \big] = \EE \big[ \sum_{uv \in P} r_{uv} \cdot X_u X_v \big],\]
where $P$ is the adaptively-chosen probe set.

We note that \bayes captures a version of Bayesian active search with a particular correlation model and Bayesian update. The edges of $G$ correspond to the items to be tested, and probing an edge corresponds to testing whether that item belongs to the target set. In our model, each item's status is determined by two ``features'' -- its endpoints (e.g. each chemical compound may have two ingredients in a drug discovery setting). An item is in the target set if and only if both its corresponding vertices are active. Since edges can share endpoints, probing one edge reveals information about its vertices, thereby updating the expected reward of all other edges incident to those same vertices.

\begin{example}\label{ex:application}
    For a concrete example, consider a network of financial entities (vertices) and the transactions between them (edges). Each entity has a certain probability of being ``corrupt" (active) or ``legitimate" (inactive). A transaction is flagged as fraudulent (yielding a reward) only if both participating entities are corrupt. An investigator can spend resources to audit a specific transaction (probe an edge). This audit reveals whether the transaction is fraudulent and, crucially, also confirms the status (corrupt or legitimate) of the two entities involved. The investigator's goal is to use a limited budget to adaptively audit transactions and uncover the maximum possible amount of fraud.  
\end{example}

For the remainder of this paper, we focus on designing algorithms for \bayes. One can compute the optimal adaptive policy -- which can be represented by a decision tree -- via a  dynamic program, but unfortunately computing this policy and its representation may require exponential time/space. Thus, we focus on efficient \emph{approximation algorithms}. For any $\alpha \geq 1$, we say an algorithm, $\alg$, is an \emph{$\alpha$-approximation} for \bayes if for any instance $\mathcal{I}$, we have $\frac{\EE[\opt(\mathcal{I})]}{\EE[\alg(\mathcal{I})]} \leq \alpha$, where $\opt$ is the reward of the optimal \emph{adaptive} policy.

We observe that the deterministic version of the problem, where $p_v = 1$ for all $v \in V$ is exactly the classic knapsack problem, which admits a Fully Polynomial Time Approximation Scheme (FPTAS) \cite{DBLP:journals/jacm/IbarraK75, DBLP:journals/mor/Lawler79}, so for any $\epsilon \in (0,1)$, we can obtain a $\frac{1}{1-\epsilon}$-approximate solution in polynomial time in the input size and $\frac{1}{\epsilon}$. However, beyond this special case, there is no non-trivial approximation result for \bayes.

\subsection{Applications and related work}

Beyond the Bayesian active search problem described above, \bayes is related to many other directions in stochastic combinatorial optimization.

\paragraph{Stochastic probing}

Bayesian probing is closely connected with stochastic probing. In stochastic probing, we are given a collection of elements $E$, such that each element is active \emph{independently} with known probability. We can adaptively probe elements to learn if they are active or not. Further, we are given a set function $f: 2^E \xrightarrow{} \mathbb{R}$. The goal of stochastic probing is to adaptively probe elements, say $P \subset E$, and choose a subset $Q \subset P$ of the probe set to maximize $\EE[f(Q \cap A)]$, where $A \subset E$ is the random subset of active elements. Stochastic probing has been studied for different $f$ (linear, submodular, XOS) and constraints on the $P$ and chosen subset $Q$ \cite{DBLP:conf/ipco/GuptaN13, DBLP:conf/soda/GuptaNS16, DBLP:conf/soda/GuptaNS17}.

The main distinction is that in \bayes, we remove the assumption that the elements are active independently. There has been some prior work on probing problems with correlated elements. However, the results rely on strong assumptions on the correlation model and function $f$. Most relevant to our work is the work of \cite{DBLP:journals/jair/GolovinK11}, where they introduce a property called \emph{adaptive submodularity} (\Cref{definition:adaptivesubmodularity}), which enables simple greedy-type $O(1)$-approximations for a variety of stochastic probing problems. We will show that \bayes does \emph{not} satisfy adaptive submodularity (\Cref{appendix: adapt}), so our work is not subsumed by theirs.

\paragraph{Stochastic knapsack}

Stochastic knapsack is a well-studied generalization of the classic knapsack problem. In stochastic knapsack, we are given a deterministic knapsack budget, say $B$, but the items are stochastic. There are $O(1)$-approximations known for multiple variants of the problem, the simplest being that each item has a random size with known distribution but deterministic reward \cite{DBLP:journals/mor/DeanGV08}. This has been generalized to random sizes with random rewards such that the size and reward distributions of an item may be arbitrarily correlated \cite{DBLP:conf/focs/GuptaKMR11}, and even further such that each item is a Markov chain \cite{DBLP:journals/mor/Ma18}. However, in all of these variations, it is assumed that the size/reward distributions are independent \emph{across items}.

In contrast, \bayes is another variation of stochastic knapsack where instead the sizes are deterministic, but the rewards are random and correlated with each other. To the best of our knowledge, we give the first positive results for stochastic knapsack with correlations across items for any correlation model.

\paragraph{Algorithmic site percolation}

Site percolation is a model of random graphs studied in probabilistic combinatorics. Given a graph $G$ and probability parameter $p \in [0, 1]$, let $G[V_p]$ be the random induced subgraph on $V_p \subset V(G)$, where each vertex $v \in V(G)$ is in $V_p$ independently with probability $p$.

Researchers have focused on proving threshold results for graph properties of the percolated graph, $G[V_p]$, such as for what range of $p$ does $G[V_p]$ contains a linear-sized component \cite{krivelevich2015phase, 715a160d-b084-3cf9-880e-c54074a10a63} or properties of particular families of graphs such as the grid graph \cite{REIDYS20093113}. Further, percolation theory has been applied to model and predict the spread of diseases \cite{davis2008abundance} and the extinction of species \cite{boswell1998habitat}.

Given a graph $G$ and probability parameter $p$, consider the Bayesian probing instance where $p_v = p$ for all $v \in V(G)$. The edges which give non-zero reward are exactly those that are in the percolated graph $G[V_p]$. Thus, \bayes captures an algorithmic question on the percolated graph: How much of the percolated graph can we discover given a budget of edge queries?

This question is of particular relevance to the above applications, where it is expensive to query the underlying graph, which involves expensive diagnoses or field work, respectively.

\paragraph{Stochastic optimization with correlations}

Another common theme throughout stochastic combinatorial optimization is the independence assumption: We typically assume that the jobs to be scheduled, items to be packed, or elements to be probed are independent of one another. There has been recent interest in relaxing this assumption by allowing correlations in problems such as Pandora's box \cite{chawla2020pandora}, mechanism design \cite {cai2021simple}, and prophet inequality \cite{livanos2024improved}. Our work on \bayes adds stochastic knapsack and probing to this list.

\subsection{Our results}

Our main question is how well can efficiently-computable policies approximate the optimal adaptive policy for \bayes?

First, we consider the simplest class of policies: \emph{non-adaptive policies}. A non-adaptive policy fixes a probe set $P \subset E$ up-front without any knowledge of the probe outcomes. In other words, a non-adaptive policy performs no Bayesian updates. Our first result gives an upper bound on the approximation ratio achieved by non-adaptive policies.

\begin{restatable}{theorem}{thmnonadapt}
	\label{thm-nonadapt}
	There exists a poly-time algorithm for \bayes that outputs a non-adaptive policy that $O(d_{max} \cdot S)$-approximates the optimal adaptive policy, where $d_{max}$ is the max degree of $G$ and $S$ is the ratio between the maximum and minimum edge sizes in $G$.
\end{restatable}

Observe that the problem of finding the optimal non-adaptive policy is exactly the classic knapsack problem, which has an $O(1)$-approximation, e.g. \cite{DBLP:journals/jacm/IbarraK75, DBLP:journals/mor/Lawler79}. Thus, the interesting aspect of \Cref{thm-nonadapt} is that the optimal non-adaptive policy $O(d_{max} \cdot S)$-approximates the optimal \emph{adaptive policy}. In other words, we upper bound the \emph{adaptivity gap} of \bayes by $O(d_{max} \cdot S)$. We define the \emph{adaptivity gap} for \bayes by $\max_{\mathcal{I}} \frac{\EE[\opt_A(\mathcal{I})]}{\EE[\opt_{NA}(\mathcal{I})]}$, where $\opt_A$ and $\opt_{NA}$ are the reward obtained by the optimal adaptive and non-adaptive policies, respectively.

Further, we show that our upper bound is asympotically best possible by lower bounding the \emph{adaptivity gap} of \bayes.

\begin{restatable}{lem}{lemadapgap}
	\label{lem-adap-gap}
	The adaptivity gap of \bayes is $\Omega(d_{max} \cdot S)$.
\end{restatable}

In terms of the size of the graph $G$ itself, we have a different adaptivity gap lower bound.

\begin{restatable}{lem}{lemmatchinglower}
    \label{lem_matching_lower}
    The adaptivity gap of \bayes is $\Omega(\lvert V \rvert^2)$ or $\Omega(\lvert E \rvert)$.
\end{restatable}

Thus, to improve over our first algorithm, we need to use adaptivity. We consider a class of adaptive policies that have two phases: an exploration phase and an exploitation phase. Roughly speaking, in the exploration phase we first use half of the overall budget to learn information (which vertices are active or not). Then in the exploitation phase, we use the remaining half of the budget to gain reward conditioned on what we learned. See \Cref{def-ee} for a formal definition of such policies. We show that this allows us to improve $d_{max}$-dependence for non-adaptive policies to $dens(G) = \max_{H \subset V(G)} \frac{\lvert E[H] \rvert}{\lvert H \rvert}$ -- the density of the densest subgraph of $G$.

\begin{restatable}{theorem}{thmandadap}\label{thm-and-adap}
	There exists a poly-time algorithm for \bayes that computes an adaptive policy that $O(dens(G) \cdot S^2)$-approximates the optimal adaptive policy, where $dens(G) = \max_{H \subset V(G)} \frac{\lvert E[H] \rvert}{\lvert H \rvert}$ is the density of the densest subgraph of $G$ and $S$ is the ratio between the maximum and minimum edge sizes in $G$.
\end{restatable}

One special case of particular interest is Bayesian active search, where every edge has size $1$. In this case, we obtain a $O(d_{max})$-approximation by a non-adaptive policy, but a $O(dens(G))$-approximation with an \emph{efficiently computable} adaptive policy. These are the first provable approximation guarantees for Bayesian active search, and we hope that our work showcases the potential of applying algorithmic and modeling techniques from stochastic combinatorial optimization to problems in Bayesian machine learning.

\subsection{Technical overview}\label{sec:overview}

We highlight our main ideas with a running example.

\begin{example}\label{ex:gap_uniform}
    Let $d \in \NN$. Then let $G$ be a collection $d$ stars, each with $d$ leaves. The center of each star is active with probability $\frac{1}{d}$ and each leaf with probability $1$. Further, each edge weight and size is $1$, and our budget is $d$. See \Cref{fig:star_ex_uni}. This completes the description of the Bayesian probing instance.
    
    Thus, our goal in this instance is to probe at most $d$ edges to maximize the expected number of active edges (those with both endpoints active) probed.
\end{example}

    \begin{figure}[h]
        \centering
        \includegraphics[scale = 0.2]{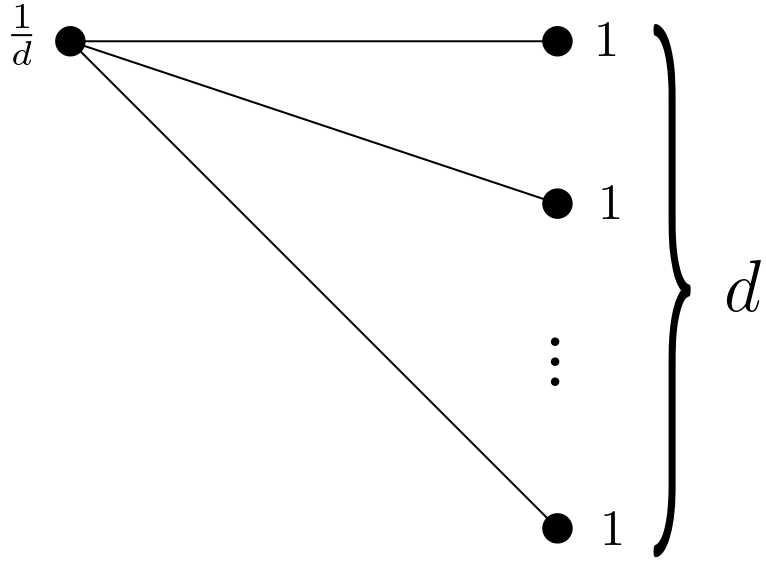}
        \caption{A single star in $G$. Vertices are labeled with their probabilities. All edges have $w_e = 1$ and $s_e = 1$.}
        \label{fig:star_ex_uni}
    \end{figure}

    First, consider probing any edge of a star. Doing so, we learn the realized value of its center. If its center is active, then we now know all edges in this star are active and will give reward $1$. Otherwise, every such edge gives reward $0$. In contrast, prior to probing this edge, every edge of this star gives expected reward $\frac{1}{d}$.

    Interestingly, such stars correspond exactly to \emph{clumps} in the Bayesian active search literature. In the language of Bayesian active search, a clump is a subcollection of items such that either all items in the clump are active or all are inactive (with some probability). Collections of clumps were introduced in \cite{GarnettKXSM12} as a hard instance for Bayesian active search. Thus, the algorithms we will design for \bayes can provably perform well on hard benchmarks for Bayesian active search.

    Going back to our particular instance in \Cref{ex:gap_uniform}, an adaptive policy on this instance can probe one edge per star until it finds a star with an active center. Then, it uses its remaining budget to probe edges in this particular star. The adaptive policy finds an active center within its first $\frac{d}{2}$ probes with at least constant probability, and -- conditioned on this -- probes at least $\frac{d}{2}$ edges that each give reward $1$. This policy has expected reward $\Omega(d)$.

    On the other hand, a non-adaptive policy can only fix a set of $d$ edges to probe in advance, each of which gives expected reward $\frac{1}{d}$. Thus, any non-adaptive policy has expected reward at most $1$, so the adaptivity gap of this example is $\Omega(d) = \Omega(d_{max})$.

    We can see in this example that the benefit of adaptivity is that probing edge immediately tells us whether all $d$ edges in the same star have reward $0$ or $1$. Indeed, our lower bound guaranteed by \Cref{lem-adap-gap} is a slight modification of the above example that takes advantage of the edge sizes and weights as well.

    Further, this example highlights the key difference between \bayes and other stochastic combinatorial optimization problems: In \bayes \emph{the reward obtained from an edge is not independent of the decision to probe that edge}. A common approach in other stochastic combinatorial optimization problems is to define a linear programming (LP) relaxation of the problem and show that the marginals of the optimal adaptive policy certify that the relaxation is a good one. Let's try this approach for Bayesian probing. The natural LP relaxation for \bayes is exactly the knapsack LP
    \begin{equation}\label{lp_knapsack}
		\lp(G, B) = \max_x \{\sum_{e \in E(G)} \EE[R_e] \cdot x_e \mid \sum_{e \in E(G)} s_e \cdot x_e \leq B,\quad x_e = 0 \text{ if } s_e > B ,\quad 0 \leq x \leq 1\}.
	\end{equation}
    Note that $\lp(G,B)$ exactly models the problem of finding the optimal \emph{non-adaptive} policy. Throughout this paper, we use $\lp(G,B)$ to denote the LP itself and its optimal objective value. Similarly, we use $\opt(G,B)$ to denote the optimal adaptive policy and its reward. Consider setting $x_e = \Pr(e \in Q)$, where the event $e \in Q$ is with respect to the optimal adaptive policy $\opt(G,B)$. This solution is feasible for $\lp(G,B)$ and achieves objective value $\sum_{e \in E(G)} \EE[R_e] \cdot \Pr(e \in Q)$. On the instance in \Cref{ex:gap_uniform}, this is $\sum_{e \in E(G)} \frac{1}{d} \cdot \Pr(e \in Q) \leq 1$, whereas $\EE[\opt(G,B)] = \Omega(d)$, so this LP is \emph{not} a good benchmark. The issue is that $\EE[R_e] \cdot \Pr(e \in Q) \neq \EE[R_e \cdot \ind_{e \in Q}]$ since for many edges, the optimal adaptive policy queries those edges only if it \emph{knows} that it will give reward $1$.

    To overcome this issue, instead of strengthening the LP by adding more constraints -- perhaps surprisingly -- we actually \emph{simplify} the LP by restricting it to a simpler subgraph where the LP is a good benchmark. This is in contrast to prior work on generalizations of stochastic knapsack, where a stronger time-indexed LP is used \cite{DBLP:conf/focs/GuptaKMR11, DBLP:journals/mor/Ma18}. Suppose that $G$ is a matching (no two edges share an endpoint). Now probing an edge gives no information about other edges, so the event $\{e \in Q\}$ is independent of the reward $R_e$. This allows us to carry out the idealized calculation above to obtain $LP(G,B) \geq \sum_{e \in E(G)} \EE[R_e] \cdot \Pr(e \in Q) = \sum_{e \in E(G)} \EE[R_e \cdot \ind_{e \in Q}] = \EE[\opt(G,B)]$.

    Our main technical tool is a \emph{decomposition lemma} that allows us to decompose the input graph $G$ into simpler graphs -- say, matchings -- and relate $\opt(G,B)$ to $\opt$ on these simpler graphs (up to a multiplicative violation of the budget).

    \begin{restatable}[Decomposition Lemma]{lem}{lemdecomp}\label{lem-decomp}
    	Given graph $G$, let $G_1, \dots, G_\ell$ be graphs on the same vertex set as $G$ such that $E(G_1) \cup \dots \cup E(G_{\ell}) = E(G)$. Then we have $\EE[\opt(G_1, 2S \cdot B)] + \dots + \EE[\opt(G_\ell, 2S \cdot B)] \geq \EE[\opt(G,B)]$ for any budget $B$.
    \end{restatable}

    Thus, to design our non-adaptive policy, we decompose the input graph $G$ into $O(d_{max})$-many matchings. On some such matching, say $M \subset G$, by averaging we have $LP(M, 2S \cdot B) \geq \EE[\opt(M,2S \cdot B)] = \Omega(\frac{1}{d_{max}}) \cdot \EE[\opt(G,B)]$. We use $LP(M, 2S \cdot B)$ (up to a budget rescaling argument) as the benchmark for our algorithm, which will be a standard knapsack LP rounding algorithm. This gives the $O(S \cdot d_{max})$-approximation guaranteed by \Cref{thm-nonadapt}.

    How can we improve on this with an adaptive policy? Let's go back to \Cref{ex:gap_uniform}, and slightly modify the adaptive policy that achieves expected reward $\Omega(1)$: First, we pick $\frac{d}{2}$ arbitrary stars and probe one edge per star. Then, \emph{conditioned on the outcomes of those probes}, we non-adaptively probe the $\frac{d}{2}$ edges with the largest expected reward. As before, with at least constant probability, in the first phase, we find an edge that gives reward $1$ and thus a center that is active. In this case, in the second phase, we are guaranteed to probe $\frac{d}{2}$ edges that each give reward $1$. We define \emph{explore-exploit (\ee)} policies to capture this type of policy.

    \begin{restatable}[Explore-Exploit (\ee) Policy]{definition}{defee}\label{def-ee}
        For any $B > 0$, an explore-exploit policy with budget $B$ on graph $G = (V,E)$  has two phases, an exploration phase and exploitation phase.
    	\begin{itemize}
    		\item \textbf{(Exploration}) First, the policy has budget $\frac{B}{2}$ for adaptive \emph{vertex probes}. In each vertex probe the policy chooses a vertex $v \in V$, and learns the realized value of $X_v$ (obtaining no reward). We say the policy \emph{explores} vertex $v$. The cost of this probe is $s_v = \min_{uv \in E} s_{uv}$.
    		\item \textbf{(Exploitation)} After the policy has explored $P_V \subset V$ and learned the values of $X_v$ for all $v \in P_V$ -- conditioned on these outcomes -- the policy \emph{non-adaptively} probes an \emph{exploitation set} $P_E \subset E$ of size at most $\frac{B}{2}$ (with respect to the edge sizes $s_e$ for $e \in E$).
    	\end{itemize}
    	The total reward of such a policy is the rewards of all edges probes in the exploitation phase, $\EE[\sum_{e \in P_E} R_e]$.
    \end{restatable}

    It turns out that \ee policies are sufficient\footnote{up to the standard knapsack issue of needing to take the single best item sometimes} to obtain a $O(S^2 \cdot dens(G))$-approximation for \bayes, as in \Cref{ex:gap_uniform}. To achieve this, we again simplify our LP benchmark to make it more amenable to analysis. Instead of restricting to a matching as in the non-adaptive case, now we restrict to a collection of stars and condition on the outcome of the exploration phase.
    
    In a collection of stars, conditioning on the exploration phase has a nice structure. Consider exploring the center vertex of some star. If it is active, then the rewards of all incident edges are now independent (since their rewards are completely determined by their leaf values). If it is inactive, then all incident edges give reward $0$ and can be ignored.
    
    Using this idea, an analogous calculation to the matching case shows show that the knapsack LP conditioned on the exploration phase gives a good LP relaxation for the reward obtained in the subsequent exploitation phase. Thus, we view the exploration phase as adaptively exploring vertices to maximize the knapsack LP value conditioned on the outcomes of the explored vertices. Since vertices are active independently, the upshot of this view is that we can leverage techniques for stochastic probing with \emph{independent} items to efficiently find a good exploration set. 

    To summarize, in a collection of stars, our adaptive \ee policy will use the first half of the budget to maximize a conditional knapsack LP benchmark. Then we will round the resulting knapsack LP to pick the edges to exploit with the remaining half of the budget. To lift the above algorithm on collections of stars to general graphs, we again use the decomposition lemma to decompose an arbitrary graph $G$ into $O(dens(G))$-many collections of stars.

	\section{Non-adaptive policies}\label{sec-nonadapt}

Our goals in this section are to show that the adaptivity gap for \bayes is $\Theta(d_{max} \cdot S)$ and to give a non-adaptive algorithm that achieves this approximation ratio. We first give our algorithm (\Cref{thm-nonadapt}) and then prove a matching lower bound on the adaptivity gap (\Cref{lem-adap-gap}).

\subsection{Non-Adaptive algorithm}

First, we observe that the problem of finding a non-adaptive policy that approximates the optimal \emph{non-adaptive} policy is just the classic knapsack problem: Our goal is to choose a subset of edges $Q$ of total size at most $B$ to maximize $\sum_{e \in Q} \EE[R_e]$. For this problem a fully polynomial time approximation scheme (FPTAS) is known, so we can obtain a $(1+ \epsilon)$-approximate solution in time $poly(n, \frac{1}{\epsilon})$ \cite{DBLP:journals/jacm/IbarraK75, DBLP:journals/mor/Lawler79}. However, it will be convenient for our analysis here and when we later design adaptive policies to use the standard $2$-approximate greedy algorithm.

\medskip
\hrule
\smallskip
\hrule
\medskip
\knapsack: Given input graph $G$ to \bayes with budget $B > 0$:
Pick the best of the following two solutions:
\begin{itemize}
	\item Sort the edges in decreasing order of \emph{value density}: $\frac{\EE[R_e]}{s_e}$. Probe edges in this order until we use up the budget.
	\item Probe the single edge with largest expected reward.
\end{itemize}
\hrule
\smallskip
\hrule
\medskip

It is known that this algorithm has the following LP-based performance guarantee, see e.g. Exercise 3.1 in \cite{williamson2011design}:

\begin{prop}\label{prop_knapsack}
	Given input graph $G$ and budget $B > 0$, \knapsack outputs a solution with reward satisfying $\knapsack(G,B) \geq \frac{1}{2} \cdot \lp(G,B)$.
\end{prop}

Suppose we run \knapsack to output a non-adaptive policy that obtains expected reward at least $\frac{1}{2} \lp(G,B)$. However, as we saw in the technical overview (\Cref{sec:overview}), $\lp(G,B)$ is \emph{not} a good relaxation for $\EE[\opt(G,B)]$ in general, but it is a good relaxation if $G$ is a matching. Our strategy to overcome this is to choose an appropriate matching, say $M \subset G$, so we have $\lp(G,B) \geq \lp(M,B) \geq \EE[\opt(M,B)]$. It remains to relate back $\EE[\opt(M,B)]$ with $\EE[\opt(G,B)]$.

The challenge here is that the edges in $E(G) \setminus E(M)$ not only allow $\opt(G,B)$ to obtain more reward, but also more information, which allows $\opt(G,B)$ to better exploit the edges in $E(M)$. To bound this discrepancy in both reward and information, we use our decomposition lemma.

\lemdecomp*

This allows us, by averaging, to choose the matching $M$ such that $\EE[\opt(M,2S \cdot B)] = \Omega(\frac{1}{d_{max}}) \cdot \EE[\opt(G,B)]$. Then to complete the proof, we argue that re-scaling down the budget by $2S$ scales the LP solution value by the same factor.

In the remainder of this section, we prove the decomposition lemma and then \Cref{thm-nonadapt}.

\subsection{Proof of decomposition lemma: \Cref{lem-decomp}}

First, we prove the decomposition lemma. The proof of the decomposition lemma follows by applying the next lemma to each of the graphs $G_1, \dots, G_\ell$, which states that by increasing our budget by a $2S$-factor, we can simulate an adaptive policy from $G$ on a subgraph $H$ of $G$. The main idea is to simulate the information gained from probing edges in $E(G) \setminus E(H)$ by instead probing incident edges in $E(H)$. These edges may have size at most a $S$-factor larger than their counterparts in $E(G) \setminus E(H)$, which leads to the increased budget in the lemma.

\begin{lem}\label{lem-decomp-sub}
	Given graphs $G, H$ on the same vertex set such that $E(H) \subset E(G)$ and budget $B$, there exists a policy on graph $H$ with budget $2S \cdot B$ achieving expected reward at least $\EE[\sum_{e \in P \cap E(H)} R_e]$, where $P$ is the adaptively-chosen set of edges probed by policy $\opt(G,B)$.
\end{lem}
\begin{proof}
	We prove the lemma by modifying the Bayesian probing problem and the initial policy given by $\opt(G,B)$. First, we restrict $\opt(G,B)$ so it can only obtain rewards from edges in $E(H)$. To do this, consider the graph $G'$ with differs from $G$ only on its edge weights, which we denote by $w'_e$ for $G'$. For edges $e \in E(H)$, we set $w'_e = w_e$. For all other edges, $e \in E(G) \setminus E(H)$, we set $w'_e = 0$. This completes the description of $G'$.
	
	We claim that $\EE[\opt(G',B)] \geq \EE[\sum_{e \in P \cap E(H)} R_e]$. To see this, note that a feasible policy in $G'$ is to follow the probes of $\opt(G,B)$; we obtain the same information from each probe, but we only obtain rewards from edges in $E(H)$.
	
	It remains to show that by multiplying the budget by $2S$, we can modify $\opt(G',B)$ such that we only probe edges in $E(H)$, so we do not even need information from edges in $E(G) \setminus E(H)$. To do so, consider the following policy on $H$: First, for every vertex $v \in V(G)$ that is isolated in $H$ but not in $G$ we sample $X_v' \sim X_v$ independently. Now if $\opt(G',B)$ probes an edge in $E(H)$, then we also probe the same edge, obtaining the same information and reward at the same cost.
 
    Otherwise $\opt(G', B)$ probes an edge in $E(G) \setminus E(H)$. Note that $\opt(G',B)$ only obtains information about the endpoints of this edge but no reward. For each endpoint, say $v$, there are two cases. If $v$ is incident to some edge in $H$, then we probe the cheapest such edge in $H$. Note that this costs at most $S$-times more than the corresponding edge probed by $\opt(G',B)$ Otherwise $v$ is isolated in $H$ but not in $G$. Then we do no probe, but we observe $X_v'$. Considering both endpoints, we learn the values of both endpoints at cost at most $2S$-times the corresponding edge probe by $\opt(G', B)$.
    
    Note that in our policy, some of our information is simulated if there are isolated in $H$, but such vertices do not affect the reward of our policy or $\opt(G',B)$ -- they only affect the distribution of the chosen edge probes of the latter, which we simulate with an identical distribution. To summarize, we have constructed an adaptive policy on $H$ which certifies $\EE[\opt(H, 2S \cdot B)] \geq \EE[\opt(G', B)] \geq \EE[\sum_{e \in P \cap E(H)} R_e]$, as required.
\end{proof}

The decomposition lemma follows quickly by applying the above lemma to each subgraph $G_1, \dots, G_\ell$.

\begin{proof}[Proof of \Cref{lem-decomp}]
	Applying \Cref{lem-decomp-sub} to each graph $G_1, \dots, G_\ell$, we have $\EE[\opt(G_i, 2S \cdot B)] \geq \EE[\sum_{e \in P \cap E(G_i)} R_e]$ for all $i = 1, \dots, \ell$. Summing over all $i$ and using the fact $E(G_1) \cup \dots \cup E(G_\ell) = E(G)$ gives the desired result:
	\[\sum_i \EE[\opt(G_i, 2S \cdot B)] \geq \sum_i \EE[\sum_{e \in P \cap E(G_i)} R_e] \geq \EE[\sum_{e \in Q} R_e]= \opt(G,B).\]
\end{proof}

\subsection{Analysis of \knapsack: \Cref{thm-nonadapt}}
Now using the decomposition lemma, we complete the proof of \Cref{thm-nonadapt} by showing \knapsack has the desired $O(d_{max} \cdot S)$ approximation ratio.

\thmnonadapt*
\begin{proof}
	First, consider a proper edge coloring of $G$ using $\ell = d_{max} + m_{max} \leq 2 d_{max}$ colors, where $m_{max}$ is the maximum number of parallel edges incident on any vertex. Each color class forms a matching. The existence of such an edge coloring is guaranteed by \cite{vizingmulti}. Let $G_1, \dots, G_\ell$ be the subgraphs of $G$ with the same vertex set as $G$ but the edges of each color class, respectively. Applying \Cref{lem-decomp} to the $G_i$'s, there exists some such subgraph, say $G_{i^*}$ such that $\EE[\opt(G_{i^*}, 2S \cdot B)] \geq \frac{1}{\ell} \EE[\opt(G,B)] \geq \frac{1}{2 d_{max}} \EE[\opt(G,B)]$.
	
	Now let $\alg$ denote the reward obtained by $\knapsack(G,B)$. Our goal is to show that $\EE[\alg] \geq \Omega(\frac{1}{d_{max} \cdot S}) \cdot \EE[\opt(G,B)]$. Applying \Cref{prop_knapsack} gives
	\[\EE[\alg] \geq \frac{1}{2} \lp(G,B) \geq \frac{1}{2} \lp(G_{i^*}, B) \geq \frac{1}{4S} \lp(G_{i^*}, 2S \cdot B),\]
	where we use the facts that $G_{i^*}$ is a subgraph of $G$ and that scaling the budget of the knapsack LP down by a $2S$-factor decreases the optimal value by at most a $2S$-factor. This is given in the following lemma, which we prove in \Cref{sec: appendix_knapsack}.

    \begin{restatable}{lem}{lemknapsackrescale}\label{lem: knapsack_rescale}
        For any graph $G$ and budgets $0 < B \leq B'$ such that every edge in $G$ has size at most $B$, we have $\lp(G,B) \geq \frac{B}{B'} \cdot \lp(G,B')$.
    \end{restatable}
	
	Next, we use the fact that $G_{i^*}$ is a matching to argue that $\lp(G_{i^*},2S \cdot B) \geq \EE[\opt(G_{i^*},2S \cdot B)]$. Consider the setting of the $x$-variables given by $x_e = \prob(e \in P)$ for all $e \in E(G_{i^*})$, where the probabilities are with respect to the policy $\opt(G_{i^*},2S \cdot B)$. Since the size of $P$ is almost surely at most $2S \cdot B$ by definition of $\opt(G_{i^*},2S \cdot B)$, $x$ is feasible for $\lp(G_{i^*},2S \cdot B)$. Further, this solution achieves objective value
	\[\sum_{e \in E(G_{i^*})} \EE[R_e] \cdot \Pr(e \in P) = \sum_{e \in E(G_{i^*})} \EE[R_e \cdot \ind_{e \in P}] = \EE[\opt(G_{i^*},2S \cdot B)],\]
	where we use the fact that $G_{i^*}$ is a matching, so the decision to probe edge $e$ is independent of its reward $R_e$ (because previous probes have revealed no information about its endpoints). Finally, we use the property $\EE[\opt(G_{i^*}, 2S \cdot B)] \geq \frac{1}{2 d_{max}} \EE[\opt(G,B)]$, we can conclude
	\[\EE[\alg] \geq \frac{1}{4S} \lp(G_{i^*}, 2S \cdot B) \geq \frac{1}{4S} \EE[\opt(G_{i^*}, 2S \cdot B)] \geq \frac{1}{8 d_{max} S} \EE[\opt(G, B)],\]
	as required.
\end{proof}

\subsection{Adaptivity gap lower bound}

Our lower bound construction that achieves \Cref{lem-adap-gap} is a slight modification of \Cref{ex:gap_uniform}. In that example, each component was a star, and further each edge size and weight was $1$. This gave an adaptivity gap lower bound of $\Omega(d_{max})$. To obtain the stronger lower bound of $\Omega(d_{max} \cdot S)$, we make one edge on each star have $0$ weight but very small size. Such an edge is useless for a non-adaptive policy since it gives no reward, but it is useful for an adaptive policy to learn if the center is active or not. This idea is formalized below.

\lemadapgap*
\begin{proof}
	Let $d \in \NN$ and $s \in (0,1)$ be parameters. Consider the graph $G$ consisting of $\frac{d}{s}$ copies of a star with $d+1$ leaves such that the center of each star has probability $\frac{s}{d}$ and every leaf has probability $1$. Each such star has one \emph{special edge} with $w_e = 0$ and $s_e = s$. The remaining edges of the star are \emph{normal edges} with $w_e = 1$ and $s_e = 1$. This completes the description of $G$. See \Cref{fig:star_ex}.
    \begin{figure}[h]
        \centering
        \includegraphics[scale = 0.2]{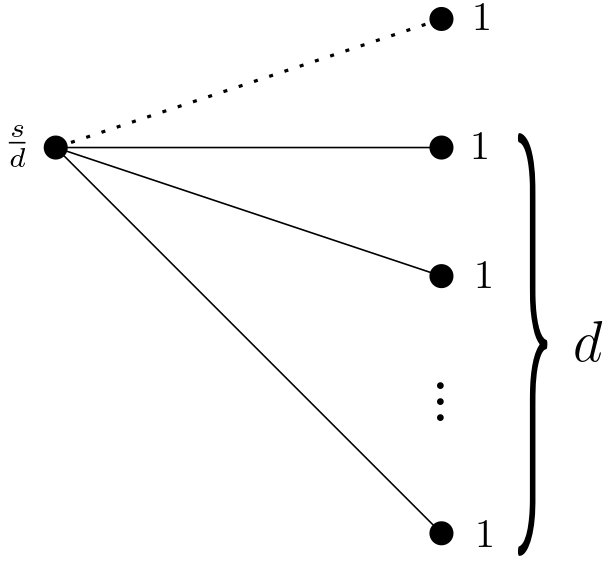}
        \caption{A single star in $G$. Vertices are labeled with their probabilities. The dotted edge is the special edge with $w_e = 0$ and $s_e = s$. The solid edges are the $d$ remaining edges with $w_e = 1$ and $s_e = 1$.}
        \label{fig:star_ex}
    \end{figure}
 
    To complete the instance, let the query budget be $d$. Note that $G$ has max degree $d + 1$ and $S = \frac{1}{s}$, so it suffices to show a gap of $\Omega(\frac{d}{s})$. Note that the special edges each have expected reward $0$ and the normal edges each have expected reward $\frac{s}{d}$, so the optimal non-adaptive policy will query $d$-many normal edges, and obtain expected reward $d \cdot \frac{s}{d} = s$.
	
	On the other hand, consider the following adaptive policy: We process the stars in arbitrary order. For each considered star, we query the special edge from that star. This gives no reward, but we learn if the center of the star is active or not. If it is active, then we use the remainder of our budget the query the remaining normal edges of this star. Otherwise, if the center is not active, then we move on to the next star. This completes the description of the adaptive policy
	
	To analyze the expected reward of this policy, observe that with constant probability, we find a star with active center within the first $\frac{d}{2s}$ queries of special edges (because we query disjoint special edges, these queried edges have independent outcomes). This costs us at most $\frac{d}{2s} \cdot s = \frac{d}{2}$ towards our budget. Conditioned on the center of this star being active, each such normal edge gives us reward $1$. In this case, we query at least $d/2$ remaining edges of this star, each of which gives reward $1$. We conclude, the expected reward of this policy is $\Omega(d)$. This gives the desired gap of $\Omega(\frac{d}{s})$.
\end{proof}
	
	\section{Adaptive policies}\label{sec-adapt}

In the previous section, we showed that the adaptivity gap for this problem is $\Theta(d_{max} \cdot S)$. Now, we show how to leverage adaptivity to go beyond this lower bound to achieve a $O(dens(G) \cdot S^2)$ by combining a non-adaptive policy with a class of limited-adaptivity policies called explore-exploit (\ee) policies. We recall our main theorem and the definition of \ee policies here

\thmandadap*

\defee*

We recall that in to \Cref{ex:gap_uniform}, where all stars are identical in number of edges, vertex probabilities, and edge sizes/weights (so $S = 1$), we argued that there exists a $O(1)$-approximate \ee policy. The above result says that this holds much more generally -- we can always find an \ee policy (or non-adaptive one) that is $O(dens(G) \cdot S^2)$-approximate for any input graph (not just a collection of identical stars).

Our proof strategy will be to use the decomposition lemma again to decompose the input graph into simpler subgraphs. In our non-adaptive algorithm, these simpler subgraphs were matchings, but here they will be \emph{collections of stars} (graphs where each connected component is a star graph). We will show later that we can decompose any graph $G$ into $O(dens(G))$-many collections of stars.

Thus, it remains to design a good adaptive policy for a collection of stars. To achieve this, we first prove some structural properties of \ee policies on collections of stars.

\subsection{Structural properties for collections of stars}

We begin with a few structural properties of adaptive policies on collections of stars, which will facilitate our algorithm and analysis. Precisely, we say $G$ is a collection of stars with center set $C$ if we can write $G = \cup_{c \in C} G_c$, where $G_c$ is a star graph with center $c$ (so we index each star by its center vertex) and leaf set $L_c$. Thus, each $G_c$ consists of its center $c$ connected individually to each of its leaves $\ell \in L_c$. Analogously as before, we define $\optee(G,B)$ to be the optimal \ee policy (and its reward) on graph $G$ with budget $B$. First we relate optimal \ee policies with typical optimal policies.

\begin{lem}\label{lem-ee-sim}
	For any graph $G$ and budget $B > 0$, we have $\optee(G, 4B) \geq \opt(G, B)$.
\end{lem}
\begin{proof}
	Consider the policy $\opt(G,B)$, which we denote by $\opt$ for convenience. We show how to simulate $\opt$ with an \ee policy $\pol$ with four times the budget. Suppose $\opt$ probes edge $uv \in E$. Then $\pol$ will explore $u$ and $v$. This ensures that $\pol$ has at least as much information as $\opt$ up until this point, and this costs $\pol$ at most $s_u + s_v \leq 2 s_{uv}$ for each such $uv$. Thus, at the end of the exploration phase, we have the same information as $\opt$ and have spent at most $2B$ units of the budget on vertex probes.
	
	Then, in the exploitation phase, we go back and -- with hindsight of all explored vertices -- we exploit exactly the same edges that $\opt$ did. Note that the exploitation phase costs at most $B$. This completes the description of $\pol$. Since we exploit the same edges as $\opt$ with the same information, we have $\pol \geq \opt$
\end{proof}

In light of the above lemma, it suffices to approximate the optimal \ee policy (up to a budget-scaling argument). Further, if $G$ is a collection of stars, then we can -- up to a constant factor in the budget -- we show that it suffices to explore the centers of the stars (as was the case in \Cref{ex:gap_uniform}). Intuitively, this should be true because exploring a leaf only helps us decide if its single incident edge is worth exploiting, where we could just exploit this edge, anyways. On the other hand, exploring centers is much more valuable, because they tell us about multiple edges at once. The next lemma formalizes this intuition via a coupling argument.

\begin{lem}\label{lem-center-query}
	For any \ee policy $\pol$ on a collection of stars with budget $B$, there exists another \ee policy $\tilde{\pol}$ with budget $2B$ such that:
	\begin{itemize}
		\item The rewards of both policies satisfies $\EE[\tilde{\pol}] \geq \EE[\pol]$.
		\item $\tilde{\pol}$ only explores centers of stars.
	\end{itemize}
\end{lem}
\begin{proof}
	We define $\tilde{\pol}$ as follows: Our goal is to simulate the exploration phase of $\pol$. In the exploration phase of $\pol$, if $\pol$ explores a center, then $\tilde{\pol}$ does as well. If $\pol$ explores a leaf, say $\ell$, then $\tilde{\pol}$ simulates exploring $\ell$ by sampling $X_\ell'$ identically distributed as $X_\ell$, and using the outcome of $X_\ell'$ to simulate the subsequent adaptive decisions of $\pol$. We say $\tilde{\pol}$ \emph{fake-explores} leaf $\ell$. This concludes the description of $\tilde{\pol}$. Note that we use at most $\frac{B}{2}$ budget in the exploration phase of $\tilde{\pol}$, since we do not pay anything for fake explorations.
	
	Next, we define the exploitation phase of $\tilde{\pol}$. Before doing that, we first introduce our coupling between $\pol$ and $\tilde{\pol}$. Let $C$ and $L  = \cup_{c \in C} L_c$ be the collections of all centers and leaves respectively. We let $(X_v)_{v \in V}$ be the realized $X$-values of the vertices for $\pol$ and $(\tilde{X}_v)_{v \in V}$ for $\tilde{\pol}$, respectively. We couple the outcomes for the centers, so $X_c = \tilde{X}_c$ for all $c \in C$. For the leaves, recall that $(X_\ell')_{\ell \in L}$ determines the simulated explorations of $\tilde{\pol}$. Here we couple $X_\ell = X_\ell'$ for all $\ell \in L$, while the $\tilde{X}_\ell$'s remain uncoupled, so they come from the same distribution as the $X_\ell$'s but are independent of the $X_\ell$'s. This concludes the description of the coupling.

    By definition of the coupling, the set of centers (and their outcomes) explored by $\pol$ and $\tilde{\pol}$ are the same per-realization. Similarly, the fake-explored leaves by $\tilde{\pol}$ are exactly those explored by $\pol$ with the same outcomes. Now let $P \subset E$ be the edges exploited by $\pol$. We partition $P = P_L \cup (P \setminus P_L)$, where $P_L$ is the subset of exploited edges such that their leaf-endpoint (note since the input graph is a collection of stars, every edge is between a center and one of its leaves) was previously explored by $\pol$. Further, define $E_L$ to be \emph{all} edges such that their leaf-endpoint was explored by $\pol$, so $P_L \subset E_L$. By our coupling, $E_L$ is also exactly the edges whose leaf-endpoint was fake-explored by $\tilde{\pol}$.
    
    Now we are ready to define the exploitation phase of $\tilde{\pol}$: in $\tilde{\pol}$, we exploit all edges $P \setminus P_L$ and $E_L$. This completes the description of $\tilde{\pol}$.
    
    First we claim that this exploitation set is feasible (i.e. it costs at most $B$). First, $P \setminus P_L$ costs at most $\frac{B}{2}$ since $\pol$ has budget $B$, and $P$ is exactly the exploitation set of $\pol$. Second, note that if $\pol$ explores a leaf, say $\ell$, it pays $s_\ell = s_{c\ell}$, where $c\ell$ is the unique center-leaf edge incident on $\ell$. Thus, $\pol$ pays at most $\frac{B}{2}$ budget exploring leaves, which is the same cost as simply \emph{exploiting} all of their incident edges, $E_L$.
    
	It remains to analyze the expected reward of $\tilde{\pol}$, which we write as
    \[ \sum_{c\ell \in E} \EE[\tilde{R}_{c\ell} \cdot \ind_{c\ell \in P \setminus P_L}] + \sum_{c\ell \in E} \EE[\tilde{R}_{c\ell} \cdot \ind_{c\ell \in E_L}],\]
    where we define $\tilde{R}_{c\ell} = r_{c\ell} \cdot \tilde{X}_c \tilde{X_\ell}$ for any $c\ell \in E$. By our coupling, we have $X_c  = \tilde{X}_c$ for every center $c$ almost surely. Further, the events $\{c\ell \in P \setminus P_L\}$ and $\{c\ell \in E_L\} \iff \{\pol \text{ explores $\ell$}\}$ are both independent of $X_\ell$ and $\tilde{X}_\ell$. For the first sum, we have
    \begin{align*}
        \sum_{c\ell \in E} \EE[\tilde{R}_{c\ell} \cdot \ind_{c\ell \in P \setminus P_L}] &= \sum_{c\ell \in E} \EE[r_{c\ell} \cdot \tilde{X}_c \tilde{X}_\ell \cdot \ind_{c\ell \in P \setminus P_L}]\\
        &= \sum_{c\ell \in E} \EE[r_{c\ell} \cdot X_c \tilde{X}_\ell \cdot \ind_{c\ell \in P \setminus P_L}]\\
        &= \sum_{c\ell \in E} \EE[r_{c\ell} \cdot X_c \cdot \ind_{c\ell \in P \setminus P_L}] \cdot  \EE[\tilde{X}_\ell]\\
        &= \sum_{c\ell \in E} \EE[r_{c\ell} \cdot X_c \cdot \ind_{c\ell \in P \setminus P_L}] \cdot  \EE[X_\ell]\\
        &= \sum_{c\ell \in E} \EE[r_{c\ell} \cdot X_c X_\ell \cdot \ind_{c\ell \in P \setminus P_L}],
    \end{align*}
    which is exactly the expected reward that $\pol$ obtains from exploiting $P \setminus P_L$. The same calculation for the second sum gives $\sum_{c\ell \in E} \EE[\tilde{R}_{c\ell} \cdot \ind_{c\ell \in P \setminus E_L}] = \sum_{c\ell \in E} \EE[r_{c\ell} \cdot X_c X_\ell \cdot \ind_{c\ell \in E_L}]$. Since $P_L \subset E_L$, we can further bound 
    \[\sum_{c\ell \in E} \EE[\tilde{R}_{c\ell} \cdot \ind_{c\ell \in P \setminus E_L}] \geq \sum_{c\ell \in E} \EE[r_{c\ell} \cdot X_c X_\ell \cdot \ind_{c\ell \in P_L}],\]
    which is exactly the expected reward that $\pol$ obtains from exploiting $P_L$. Combining our bounds for both sums gives $\EE[\tilde{\pol}] \geq \EE[\pol]$, as required.
\end{proof}

Thus, we define $\opteec(G,B)$ to be the optimal \ee policy (and its reward) on \emph{collection of stars} $G$ with budget $B$ \emph{that only explores the centers of stars}. We can compose the above two lemmas to obtain our main structural property for collections of stars.

\begin{cor}\label{cor-star-opt}
    For any collection of stars $G$ and budget $B > 0$, we have $\EE[\opteec(G, 8B)] \geq \EE[\opt(G,B)]$.
\end{cor}

\subsection{Exploration phase}

Starting from \Cref{cor-star-opt}, our goal is to understand \ee policies that only explore centers of stars. As before, let $G = \cup_{c \in C} G_c$ be a collection of stars with center set $C$. Suppose we have already explored centers $P \subset C$. How should we decide which edges to exploit (non-adaptively, conditioned on the outcomes of $P$)?

As a warm-up, assume that we will only exploit edges in the stars incident to $P$. This is because our final algorithm will randomize between an \ee policy and non adaptive one, so the non adaptive policy will be responsible for capturing the reward from edges not incident to $P$ (since our exploration phase gives us no extra information on such edges).

For the centers in $P$ which are inactive, we know that exploiting any edge in those stars will give zero reward, so we may ignore such stars. For the active centers in $P$, any incident edge, say $c\ell$ now has reward conditionally-distributed as $w_{c\ell} \cdot X_\ell$, rather than $w_{c\ell} \cdot X_c X_\ell$ before. Importantly, note that the rewards of such edges are conditionally independent of each other. We let $A \subset C$ denote the subset of active centers.

It follows -- if we are only exploiting edges incident to $P \cap A$ -- the best edges to exploit are exactly the optimal solution of the knapsack instance with one item per edge incident to $P \cap A$, say $c\ell$, with reward $w_{c\ell} \cdot \EE[X_\ell]$ and size $s_{c\ell}$. This motivates the following definition to capture conditioning on an active exploration set $A \subset C$. For any collection of stars $G = \bigcup_{c \in C} G_c$ and center set $C' \subset C$, we define $G \mid C'$ to be the subgraph $\bigcup_{c \in C'} G_c$, where we update the probabilities $p_c \rightarrow 1$ for every $c \in C'$. In other words, we condition on $C'$ being the set of active centers and only consider those stars.

In particular, the natural knapsack LP on instance $G \mid C'$ is given by
\begin{equation}\label{lp-cond}
	\lp(G \mid C', B) = \max_x \{\sum_{c\ell \in E(G \mid C')} \EE[w_{c\ell} \cdot X_\ell] \cdot x_{c\ell} \mid \sum_{e \in E(G \mid C')} s_e \cdot x_e \leq B,\quad x_e = 0 \text{ if } s_e > B ,\quad 0 \leq x \leq 1\}.
\end{equation}
In an \ee-policy with budget $B$, the exploitation phase non-adaptively probes edges of total size at most $B/2$. Thus, given that we already explored centers $P$ and observed that $P \cap A \subset C$ was the active subset, we can simply run $\knapsack(G \mid P \cap A, B/2)$ to choose a near-optimal exploitation set (again, recalling that at this point we only want to get reward from active probed stars). This gives expected conditional reward $\Omega \big( \lp(G \mid P \cap A, B/2) \big)$.

The upshot of the above argument is that we can view the exploration phase as adaptively choosing a set of centers to probe $P \subset C$ to maximize $\EE[\lp(G \mid P \cap A, B/2)]$, where $A \subset C$ is the subset of active centers. We will show that this function has a useful property called \emph{adaptive submodularity} that makes it easy to optimize adaptively. Informally, adaptive submodularity means that the expected marginal gain of probing an element only decreases as we probe more elements conditioned on any partial realization of the past probes.

\paragraph{Adaptive submodular maximization}

We set up the preliminaries needed to define adaptive submodularity. Suppose we have a ground set $Y$ of elements such that each element $y \in Y$ is in one of two possible states: inactive or active. Further, we have some probability distribution over the subset of active elements, say $\mathcal{A}$ on $2^Y$,

Now let $f: 2^Y \times 2^Y$ be a function. We interpret the first argument as the choice of our algorithm and the second as the realization of the active elements. Looking ahead, our algorithm will adaptively probe elements, observe if the probed element is active or not, and then probe the next element. We introduce \emph{partial realizations} of active elements to capture the partial information revealed by our algorithm. 

A partial realization of a subset $Y' \subset Y$ is a subset $A_{Y'} \subset Y'$, which specifies the active elements in $Y'$. For subsets $Y_1 \subset Y_2 \subset Y$ with partial realizations $A_{Y_1}$ and $A_{Y_2}$, respectively, we say these partial realizations are \emph{consistent} if $A_{Y_2} \cap Y_1 = A_{Y_1}$ (i.e. if they agree on the active elements in the smaller subset).

Then we define, the marginal gain of adding element $y \in Y$ to a subset $Y' \subset Y$ of already-probed elements with partial realization $A_{Y'}$ by
\begin{equation}\label{eq: gain}
  \gain(y, Y' \mid A_{Y'}) = \EE_{A \sim \mathcal{A}}[f(Y' \cup \{y\}, A) \mid A \cap Y' = A_{Y'}] - \EE_{A \sim \mathcal{A}}[f(Y', A) \mid A \cap Y' = A_{Y'}].  
\end{equation}
We are now ready to define adaptive submodularity.

\begin{restatable}[Adaptive Submodularity]{definition}{adaptivesubmodularity}\label{definition:adaptivesubmodularity}
    The function $f: 2^Y \times 2^Y \rightarrow \mathbb{R}_{\geq 0}$ is adaptive submodular (with respect to distribution $\mathcal{A}$ on $2^Y$) if for any subsets $Y_1 \subset Y_2 \subset Y$, $y \notin Y_2$, and consistent partial realizations $A_{Y_1}$ and $A_{Y_2}$  of $Y_1$ and $Y_2$, respectively, we have
    \[\gain(y, Y_1 \mid A_{Y_1}) \geq \gain(y, Y_2 \mid A_{Y_2}),\]
    where $\gain(\cdot, \cdot \mid \cdot)$ is defined as in (\ref{eq: gain}).
\end{restatable}

We note that one can define adaptive submodularity for more general state spaces (i.e. not only elements being active or inactive), but we chose to specialize our definition to prevent excess notation. See \cite{DBLP:journals/jair/GolovinK11} for a more general definition.

Adaptive submodularity is a useful property, which enables greedy-type algorithms to achieve good performance guarantees for a variety of optimization problems \cite{DBLP:journals/jair/GolovinK11, hellerstein2021tight, esfandiari2021adaptivity, tang2022beyond, cui2023minimum}. Of particular interest to us is the problem of maximizing an adaptive submodular function subject to a knapsack constraint (\adaptiveknapsack).

\begin{definition}[\adaptiveknapsack]\label{definition:adaptiveknapsack}
    In an instance of \adaptiveknapsack, we are given an adaptive submodular function $f: 2^Y \times 2^Y \rightarrow \mathbb{R}_{\geq 0}$ with respect to distribution $\mathcal{A}$ on $2^Y$, element costs $c_y \geq 0$ for $y \in Y$, and a budget $B \geq 0$. The goal is to adaptive probe elements $P \subset Y$ of total cost at most $B$ to maximize $\EE_{A \sim \mathcal{A}, P}[f(P,A)]$.
\end{definition}

This problem admits a $O(1)$-approximation by a natural adaptive greedy algorithm with respect to value density (marginal gain of adding an item divided by its cost) along with an edge case of picking the single best item. This is given in the next proposition due to \cite{amanatidis2020fast}.

\begin{prop}\label{prop:adaptiveknapsack}
    Given an oracle that computes $\gain(\cdot, \cdot \mid \cdot)$ for any inputs in polynomial time, there exists a polynomial time algorithm for \adaptiveknapsack that adaptively probes $P \subset Y$ of total cost at most $B$ such that $\EE[f(P, A)] \geq \frac{1}{9} \cdot \EE[f(P^*,A)]$, where $P^*$ is the optimal adaptive policy. 
\end{prop}

To apply the above proposition, we show that our function of interest, $f(P,A) = \lp(G \mid P \cap A, B)$ is adaptive submodular, and we give an efficient oracle for $\gain$.

\begin{lem}\label{lem:adaptivesubmod}
    For any collection of stars $G$ with center set $C$ and budget $B \geq 0$, the function $f: 2^C \times 2^C \rightarrow \mathbb{R}_{\geq 0}$ defined by $f(P,A) = \lp(G \mid P \cap A, B)$ is adaptive submodular. Further, there exists a polynomial time algorithm that, given any $C' \subset C$, element $c \notin c'$, and partial realization $A_{C'}$ of $C$, computes $\gain(c, C', A_{C'})$.
\end{lem}
\begin{proof}
    Consider any subset $C' \subset C$ with partial realization $C'$ and any element $c \notin C'$. We can write $\gain(c, C', A_{C'})$ as
    \begin{align*}
        \gain(c, C', A_{C'}) &= \EE[f(C' \cup \{c\}, A) \mid A \cap C' = A_{C'}] - \EE[f(C', A) \mid A \cap C' = A_{C'}]\\
        &= \EE[ \lp (G \mid (C' \cup \{c\} ) \cap A, B) \mid A \cap C' = A_{C'}] - \EE[ \lp (G \mid C' \cap A, B) \mid A \cap C' = A_{C'}]\\
        &= \EE[ \lp (G \mid A_{C'} \cup (\{c\} \cap A), B) \mid A \cap C' = A_{C'}] - \EE[ \lp (G \mid A_{C'}, B) \mid A \cap C' = A_{C'}]\\
        &= \EE[ \lp (G \mid A_{C'} \cup (\{c\} \cap A), B) \mid A \cap C' = A_{C'}] - \lp (G \mid A_{C'}, B)
    \end{align*}
    Recall that our distribution over $A$ includes $c \in A$ independently for each $c \in C$ with probability $p_c$. In particular, the event $c \in A$ is independent of $A \cap C' = A_{C'}$, since $c \notin C'$. Thus, we can rewrite the first term as
    \[\EE[ \lp (G \mid A_{C'} \cup (\{c\} \cap A), B) \mid A \cap C' = A_{C'}] = p_c \cdot \lp (G \mid A_{C'} \cup \{c\}, B) + (1 - p_c) \cdot \lp (G \mid A_{C'}, B).\]
    Plugging this into $\gain$ gives our final expression
    \[\gain(c, C', A_{C'}) = p_c \cdot \left( \lp (G \mid A_{C'} \cup \{c\}, B) - \lp (G \mid A_{C'}, B) \right).\]
    To evaluate this expression, we need to solve two polynomially-sized LP's, so it is clear that $\gain(c, C', A_{C'})$ can be computed efficiently. It remains to show that $f$ is adaptive submodular. Again, using our expression for $\gain$, it suffices to show that the simpler function $g: 2^C \rightarrow \mathbb{R}_{\geq 0}$ defined by $g(C') = \lp(G \mid C', B)$ is submodular in the usual sense (for any sets $C_1 \subset C_2 \subset C$ and $c \notin C_2$, we have $g(C_1 \cup \{c\}) - g(C_1) \geq g(C_2 \cup \{c\}) - g(C_2)$. This is given in the following proposition, which completes the proof of the lemma.

    \begin{restatable}{prop}{propsubmod}\label{prop:submod}
        For any collection of stars $G$ with center set $C$ and budget $B \geq 0$, the set function $g(P) = \lp(G \mid P, B)$ for $P \subset C$ is submodular.
    \end{restatable}

    We prove the proposition formally in \Cref{sec: appendix_knapsack} using standard properties of the knapsack LP. Roughly speaking, one can view the marginal gain of adding $c$ to $C_1$ as adding some new items to the knapsack instance defined by $C_1$. These items may replace some of the items chosen by the optimal LP solution on $C_1$. Since $C_1 \subset C_2$, the items replaced by adding $c$ to $C_2$ are only better than those replaced in $C_1$, which gives the desired result.
\end{proof}

To summarize this section, we can reduce the problem of finding a good exploration set to an instance of \adaptiveknapsack with objective function $f(P,A) = \lp(G \mid P \cap A, B/2)$. In the next section, we give our complete algorithm for collections of stars

\subsection{Algorithm for collections of stars}

Morally, our algorithm to find a good \ee policy for collections of stars will use the algorithm guaranteed by \Cref{prop:adaptiveknapsack} for \adaptiveknapsack to find a good set of centers to explore. Then, conditioned on the outcome of the exploration phase, we run \knapsack to choose the exploitation set. However, there are four technical details which we have swept under the rug until now:

\begin{itemize}
	\item Our final goal is not to output an \ee policy, but rather a policy that does edge queries. It is straightforward to simulate an \ee policy with budget $B$ by typical edge queries (given information and reward) with the same budget -- if the \ee policy explores some vertex, then we can perform an edge query on the cheapest edge incident to that vertex and obtain only more information (and possibly some reward). In the exploitation phase, we can query the same edges as the \ee policy. Thus, it suffices for our algorithm to output an \ee policy.
	\item In two steps so far, we have lost a constant factor in the budget: when comparing to the optimal \ee policy rather than typical adaptive policy (\Cref{lem-ee-sim}) and when considering \ee policies that only explore centers (\Cref{lem-center-query}). To fix this, we will show that we can scale-down our budget by a constant factor at only a constant factor increase in our final approximation ratio. Here we use properties of the (conditional) knapsack LP.
	\item Although we are exploring to optimize $f(\cdot, A)$, note that $f(\cdot, A)$ only counts the reward from stars whose centers we have explored. In general, $\opteec$ can obtain reward from stars with unexplored centers as well. To fix this, we randomly choose the exploitation phase to either choose edges based conditioned on the exploitation set (to get reward from explored centers) or ignoring it (to get reward from unexplored centers).
    \item One limitation of \ee policies is that they can never obtain reward from \emph{large} edges (edges $e$ with size $s_e > B/2$). However, we will show that the adaptivity gap for the the subinstance of large edges is small, so we can simply run \knapsack if such edges contribute a large amount to the optimal reward.
\end{itemize} 

Taking the above points into account, our \ee policy, \staralg, is given below.

\medskip
\hrule
\smallskip
\hrule
\medskip
\staralg: Given input collection of stars $G$ and budget $B > 0$:
\begin{enumerate}
	\item \textbf{(Exploration)} Run the algorithm guaranteed by \Cref{prop:adaptiveknapsack} on function $f(P,A) = \lp(G \mid P \cap A, B/2)$ with budget $\frac{B}{2}$ to adaptive explore centers, say $Q \subset C$.
	\item \textbf{(Exploitation)} With probability $\frac{1}{2}$ each, run one of the two following algorithms:
    \begin{itemize}
        \item Exploit the edges output by $\knapsack(G \mid Q \cap A, B/2)$. where $A \subset C$ is the realized set of active centers.
        \item Exploit the edges output by $\knapsack(G, B/2)$.
    \end{itemize}
\end{enumerate}
\hrule
\smallskip
\hrule
\medskip

Our main guarantee for \staralg is given in the next theorem. Note that we are comparing against a stronger $\opt$ with budget $B' \geq B$. This will be useful for future budget-scaling arguments. Also, note that we restrict ourselves to instances with only small edges ($s_e \leq B/2$); we defer the large edge case to later.

\begin{theorem}\label{thm:starsmall}
    Let $G$ be a collection of stars and $B' \geq B > 0$ any budgets. If every edge $e \in E$ has size $s_e \leq B/2$, then $\staralg(G,B)$ outputs an \ee policy achieving expected reward $\EE[\staralg(G,B)] = \Omega((\frac{B}{B'})^2) \cdot \EE[\opt(G,B')]$
\end{theorem}

\subsection{Analysis of \staralg: Proof of \Cref{thm:starsmall}}

We begin by passing from standard adaptive policies ($\opt$) to \ee policies that only explore centers of stars ($\opteec$). By \Cref{cor-star-opt}, we have
\[\EE[\opteec(G, 8B')] = \Omega(\EE[\opt(G,B')]).\]
We further partition the reward obtain by $\opteec$ into the reward obtain from exploiting edges whose center we previously explored or not, which we denote by $\opt^1$ and $\opt^0$, respectively:
\[\opteec(G, 8B') = \opt^1 + \opt^0.\]
On the other hand, we can write the expected reward of our algorithm as
\[\EE[\staralg(G,B)] = \frac{1}{2} \cdot \EE[\knapsack(G \mid Q \cap A, B/2)] + \frac{1}{2} \cdot \EE[\knapsack(G, B/2)].\]
To complete the proof, we bound $\EE[\knapsack(\cup_{q \in Q \cap A} G_q, B/2)]$ with respect to $\EE[\opt^1]$ and $\EE[\knapsack(G, B/2)]$ with respect to $\EE[\opt^0]$ in order.

The former bound is more involved; we relate $\opt^1$ to a solution of the conditional LP: $\lp(G \mid Q^* \cap A, 4B')$, where $Q^* \subset C$ is the subset of explored centers chosen by $\opteec(G, 8B')$. Then we compare our algorithm's choices in both the exploration and exploitation phases to those of $\opteec(G, 8B')$ in terms of appropriate \adaptiveknapsack instances.

\begin{lem}\label{lem:explored}
    We have $\EE[\knapsack(G \mid Q \cap A, B/2)] = \Omega((\frac{B}{B'})^2) \cdot \EE[\opt^1]$.
\end{lem}
\begin{proof}
    As in the discussion above, we let $Q^*$ be the exploration set of total size at most $4B'$, which is adaptively chosen by $\opteec(G, 8B')$. We first bound the $\opt^1$ in terms of the conditional LP. Let $E^1 \subset E$ be the subset of edges exploited by $\opteec(G, 8B')$ whose center was previously explored. We may assume that $E_1 \subset \bigcup_{c \in Q^* \cap A} E(G_c)$ -- i.e. $E^1$ only contains edges incident on stars with active centers, since all edges incident to inactive centers give zero reward. For convenience, we define $G(Q^* \cap A) = \bigcup_{c \in Q^* \cap A} G_c$ to be the random subgraph of such stars with explored and active centers. Then we can write $\EE[\opt^1]$ as
    \begin{align*}
        \EE[\opt^1] &= \EE \big[ \sum_{e \in E^1} R_e \big]\\
        &= \EE[\sum_{e \in E} R_e \cdot \ind_{e \in E^1} \cdot \ind_{e \in G(Q^* \cap A)}]\\
        &= \EE[\sum_{e \in E} \EE[R_e \cdot \ind_{e \in E^1} \cdot \ind_{e \in G(Q^* \cap A)} \mid Q^* \cap A]]\\
        &= \EE[\sum_{e \in E} \ind_{e \in G(Q^* \cap A)} \cdot \EE[R_e \cdot \ind_{e \in E^1} \mid Q^* \cap A]]\\
        &= \EE[\sum_{e \in G(Q^* \cap A)} \EE[R_e \cdot \ind_{e \in E^1} \mid Q^* \cap A]] 
    \end{align*}
    Now we observe that on the event $c\ell \in E^1$, we already know that $c$ is active, so $R_{c\ell} = w_{c\ell} \cdot X_c X_\ell = w_{c\ell} X_\ell$. It follows, for any $c\ell \in G(Q^* \cap A)$, we have
    \[\EE[R_{c\ell} \cdot \ind_{c\ell \in E^1} \mid Q^* \cap A] = \EE[w_{c\ell} X_{\ell} \cdot \ind_{c\ell \in E^1} \mid Q^* \cap A] = \EE[w_{c\ell} X_{\ell}] \cdot \Pr(e \in E^1 \mid Q^* \cap A).\]
    Plugging this expression into the above gives
    \[\EE[\opt^1] = \EE[\sum_{e \in G(Q^* \cap A)} \EE[w_{c\ell} X_{\ell}] \cdot \Pr(e \in E^1 \mid Q^* \cap A)].\]
    Now consider the random LP solution $x_e = \Pr(e \in E^1 \mid Q^* \cap A)$ for all $e \in G(Q^* \cap A)$. Since $E^1$ has total size at most $4B'$ almost surely, this setting of the $x$-variables is feasible for $\lp(G \mid Q^* \cap A, 4B')$, and it achieves objective value exactly $\sum_{e \in G(Q^* \cap A)} \EE[w_{c\ell} X_{\ell}] \cdot \Pr(e \in E^1 \mid Q^* \cap A)$. We conclude that
    \[\EE[\opt^1] \leq \EE[\lp(G \mid Q^* \cap A, 4B')].\]
    Since every edge $e \in E(G)$ has size $s_e \leq B/2$, we can apply \Cref{lem: knapsack_rescale} to obtain
    \[\EE[\opt^1] = O(\frac{B'}{B}) \cdot \EE[\lp(G \mid Q^* \cap A, B/2)].\]
    On the other hand, the quantity we are comparing to is
    \[\EE[\knapsack(G \mid Q \cap A, B/2)] = \Omega(\EE[\lp(G \mid Q \cap A, B/2)]).\]
    To complete the proof of the lemma, it suffices to show $\EE[\lp(G \mid Q^* \cap A, B/2)] = O(\frac{B'}{B}) \cdot \EE[\lp(G \mid Q \cap A, B/2)]$. The main difference is that in the former expression, $\opteec(B, 8B')$ adaptively chooses $Q^*$ of total size at most $4B'$, where our algorithm chooses $Q$ of total size at most $B/2$. To do so, we leverage the submodularity of the function $\lp(G \mid \cdot, B/2)$, which was shown in \Cref{prop:submod}.

    Note that both sets $Q$ and $Q^*$ are chosen adaptively, but it will be convenient to reason about non-adaptive exploration sets instead. The following proposition bounds that \emph{adaptivity gap} for our particular instances of \adaptiveknapsack. This follows from known adaptivity gaps for certain stochastic probing problems \cite{DBLP:conf/soda/GuptaNS16}. Just like for \bayes, we define the \emph{adaptivity gap} for \adaptiveknapsack by $\max_{\mathcal{I}} \frac{\EE[\opt_A(\mathcal{I})]}{\EE[\opt_{NA}(\mathcal{I})]}$, where $\opt_A$ and $\opt_{NA}$ are the reward obtained by the optimal adaptive and non-adaptive policies, respectively. See \Cref{sec: appendix_knapsack} for proof of the next proposition, which reduces \adaptiveknapsack to an appropriate submodular probing problem and applies a known adaptivity gap result.

    \begin{restatable}{prop}{knapsackgap}\label{prop:knapsack_gap}
        Let $G$ be any collection of stars and $B > 0$ any budget. Then consider any \adaptiveknapsack instance with any budget $B' > 0$ on function $f:2^C \times 2^C \rightarrow \mathbb{R}_{\geq 0}$ defined by $f(P,A) = \lp(G \mid P \cap A, B)$ with respect to distribution $\mathcal{A}$ on $2^C$ such that for $A \sim \mathcal{A}$, each $c \in C$ is included in $A$ independently with probability $p_c \in [0,1]$. The adaptivity gap for such instances is at most $3$.
    \end{restatable}
    
    To apply \Cref{prop:knapsack_gap}, consider the following two \adaptiveknapsack instances, both on the function $f(P,A) = \lp(G \mid P \cap A , B/2)$ with underlying active item distribution defined by the probabilities $p_c$ for all $c \in C$. The instance $\mathcal{I}_{B/2}$ has budget $B/2$ and the instance $\mathcal{I}_{4B'}$ has budget $4B'$.

    We recall that our algorithm ran the $9$-approximation for \adaptiveknapsack guaranteed by \Cref{prop:adaptiveknapsack} to obtain $Q$. Thus, we have $\EE[f(Q,A)] \geq \frac{1}{9} \cdot \EE[\opt(\mathcal{I}_{B/2})]$.

    On the other hand, $Q^*$ is a feasible solution to $\mathcal{I}_{4B'}$, so $\EE[f(Q^*,A)] \leq \EE[\opt(\mathcal{I}_{4B'})] \leq 3 \cdot \EE[\opt_{NA}(\mathcal{I}_{4B'})]$, where in the final step we use the adaptivity gap (\Cref{prop:knapsack_gap}). Let $Q' \subset C$ be the optimal \emph{non-adaptive} solution to $\mathcal{I}_{4B'}$. Since every edge of $G$ has size at most $B/2$, we can partition $Q' = \bigcup_k Q'_k$ into $O(\frac{B'}{B})$-many subsets, $Q'_k$, each of which having total size in $[B/4, B/2]$. Let $g(P) = \lp(G \mid P, B/2)$. Then using the submodularity of $g$ (\Cref{prop:submod}), we have
    \[\EE[\opt_{NA}(\mathcal{I}_{4B'})] = \EE[g(Q' \cap A)] \leq \sum_k \EE[g(Q'_k 
    \cap A)].\]
    In particular, by averaging, there exists some index $k'$ with $\EE[g(Q'_{k'} 
    \cap A)] = \Omega(\frac{B}{B'}) \cdot \EE[\opt_{NA}(\mathcal{I}_{4B'})]$. Since $Q'_{k'}$ has total size at most $B/2$, it is feasible for $\mathcal{I}_{B/2}$, so we conclude that
    \[\EE[\lp(G \mid Q \cap A, B/2)] = \Omega(\EE[\opt(\mathcal{I}_{B/2})]) = \Omega(\EE[\lp(G \mid Q'_{k'} \cap A, B/2]) = \Omega(\frac{B}{B'}) \cdot \EE[\lp(G \mid Q^* \cap A, B/2)],\]
    as required.
\end{proof}

The latter bound is the simpler one. Intuitively, $\opt^0$ has no advantage over a non-adaptive algorithm, since it exploits edges without any information about either endpoint, so we can relate $\opt^0$ to the standard knapsack $\lp(G, 4B')$.

\begin{lem}\label{lem:unexplored}
    We have $\EE[\knapsack(G, B/2)] = \Omega(\frac{B}{B'}) \cdot \EE[\opt^0]$.
\end{lem}
\begin{proof}
    Let $E^0 \subset E$ be the subset of edges exploited by $\opteec(G, 8B')$ whose center was not previously explored (i.e. $E^0$ is exactly the subset of edges contributing to $\opt^0$. Then we can write $\EE[\opt^0]$ as
    \[\EE[\opt^0] = \EE \big[ \sum_{e \in E^0} R_e \big] = \sum_{e \in E} \EE[R_e \cdot \ind_{e \in E^0}] = \sum_{e \in E} \EE[R_e] \cdot \Pr(e \in E_0),\]
    where we use the fact that for any edge $e \in E$, the random variable $R_e$ is independent of the event $e \in E_0$, since we do not explore either endpoint of $e$ before deciding to exploit $e$. Thus, the marginals $x_e = \Pr(e \in E)$ for all $e \in E$ give a feasible solution to $\lp(G, 4B')$ since we have budget $4B'$ to exploit edges. This solution achieves objective value exactly $\sum_{e \in E} \EE[R_e] \cdot \Pr(e \in E_0) = \EE[\opt^0]$. We combine the guarantee for \knapsack, \Cref{prop_knapsack}, with the budget-scaling for the knapsack LP, \Cref{lem: knapsack_rescale}, to complete the proof:
    \[\EE[\knapsack(G, B/2)] = \Omega(\lp(G,B)) = \Omega(\frac{B}{B'}) \cdot \lp(G, 4B') =  \Omega(\frac{B}{B'}) \cdot \EE[\opt^0].\]
\end{proof}

We combine the above two lemmas to complete the proof of \Cref{thm:starsmall}.
\begin{align*}
    \EE[\staralg(G,B)] &= \frac{1}{2} \cdot \EE[\knapsack(G \mid Q \cap A, B/2)] + \frac{1}{2} \cdot \EE[\knapsack(G, B/2)]\\
    &= \Omega((\frac{B}{B'})^2) \cdot \EE[\opt^1] + \Omega(\frac{B}{B'}) \cdot \EE[\opt^0]\\
    &= \Omega((\frac{B}{B'})^2) \cdot \EE[\opteec(G, 8B')] = \Omega((\frac{B}{B'})^2) \cdot \EE[\opt(G, B')],
\end{align*}
as required.

\subsection{Algorithm for large edges}

The limitation of \Cref{thm:starsmall} is that we need all edges to be small (have size at most half the budget). To remedy this, in this section we show that our original non-adaptive algorithm, \knapsack, has good performance when all edges are large. Intuitively, this is because in this case, $\opt$ probes only few edges. Thus just non-adaptively taking the single best edge is a good algorithm.

\begin{lem}\label{lem:starbig}
    Let $G$ be a collection of stars and $B' \geq B > 0$ any budgets. If every edge $e \in E$ has size $s_e \in (B/2, B]$, then we have $\EE[\knapsack(G,B)] = \Omega((\frac{B}{B'})^2) \cdot \EE[\opt(G,B')]$.
\end{lem}
\begin{proof}
    We first relate $\opt$ with an \ee policy that only explores centers of stars. By \Cref{cor-star-opt}, we have $\EE[\opt(G,B')] \leq \EE[\opteec(G, 8B')]$. Now consider the following subgraph $M \subset G$: For every star $G_c$ of $G$, we keep only the single edge with largest expected reward, i.e. the edge, say $e^*_c$ achieving $\max_{c\ell \in E(G_c)} \EE[R_{c\ell}]$ for all $c \in C$. Note that $M$ is a matching. We bound the loss in expected reward by restricting $\opteec$ to $M$ in the next proposition.

    \begin{prop}\label{prop:starsmallsub}
        We have $\EE[\opteec(M,16B')] = \Omega(\frac{B}{B'}) \cdot \EE[\opteec(G,8B')].$
    \end{prop}
    \begin{proof}
        Note that $M$ is also a collection of stars with the same center set as $G$. Further, since every edge in $G$ (and thus $M$) has size in $(B/2, B]$, the cost of exploring a vertex in $M$ and $G$ is also always in $(B/2, B]$. With these observations in mind, we consider the following \ee policy $\pol$ on $M$: In the exploration phase, $\pol$ adaptively explores exactly the same centers as $G$. We couple the outcomes of both policies' exploration phase such that they explore the same centers and observe the same outcomes. Then in the exploitation phase, for each star in $G$ such that $\opteec(G,8B')$ exploits any edge in this star, our policy $\pol$ will exploit the unique edge incident to that star in $M$ (i.e. the edge with largest expected reward on that star from $G$). This completes the description of $M$.

        Note that $M$ uses budget at most $8B'$ in each phase since in both graphs, the cost of exploring/exploiting a center/edge is always in $(B/2, B]$, so $\EE[\opteec(M,16B')] \geq \EE[\pol]$. It remains to show that $\EE[\pol] = \Omega(\frac{B}{B'}) \cdot \EE[\opteec(G,8B')]$. To see this, we condition on the outcome of the exploration phase (which is the same for both policies by our coupling). Conditioned on the exploration phase, we compare the expected conditional reward obtained by $\pol$ and $\opteec(G,8B')$ per star.

        Consider any star with center $c$. Let $E_c \subset E(G_c)$ be the subset of edges exploited by $\opteec(G,8B')$ in $G_c$. There are three cases to consider:
        \begin{itemize}
            \item If $c$ is unexplored, then the expected conditional reward of $\opteec(G,8B')$ from $G_c$ is exactly $\sum_{e \in E(G_c)} \EE[R_e]$. On the other hand, $\pol$ exploits $e^*_c \in E(G_c)$ achieving expected conditional reward $\EE[R_{e^*_c}] = \Omega(\frac{B}{B'}) \cdot \sum_{e \in E(G_c)} \EE[R_e]$. This is because every edge has size at least $B/2$, $\opteec(G, 8B')$ can exploit at most $O(\frac{B'}{B})$ edges per star, and $e^*_c$ maximizes the expected reward among edges in this star.
            \item Analogously, if $c$ is explored and active, then the expected conditional reward of $\opteec(G,8B')$ from $G_c$ is exactly $\sum_{e \in E(G_c)} \EE[R_e \mid X_c = 1] = \sum_{e \in E(G_c)} \EE[w_{c\ell} X_\ell]$. Recall by definition that $e^*_c$ is chosen to maximize $\EE[R_{c\ell}] = \EE[w_{c\ell} \cdot X_c X_\ell] = w_{c\ell} \cdot p_c p_\ell$ among edges with the same center $c$. This also implies that $e^*_c$ maximizes $\EE[R_{c\ell} \mid X_c = 1] = w_{c\ell} p_\ell$ among such edges, since they all share the same center. Thus, by an analogous argument as above, we have $\EE[R_{e^*_c} \mid X_c = 1] = \Omega(\frac{B}{B'}) \cdot \sum_{e \in E(G_c)} \EE[R_e \mid X_c = 1]$.
            \item Finally, if $c$ is explored and inactive, then both policies receive zero reward from this star.
        \end{itemize}
        To conclude, for each star, the expected conditional reward received by $\pol$ from this star is at least a $\Omega(\frac{B}{B'})$-factor of the expected conditional reward received by $\opteec(G, 8B')$. Summing over all stars completes the proof.
    \end{proof}
    Using the above proposition and recalling that \ee policies are only weaker than typical adaptive policies, we have
    \[\EE[\opt(G,B')] \leq \EE[\opteec(G, 8B')] = O(\frac{B'}{B}) \cdot \EE[\opteec(M,16B')] = O(\frac{B'}{B}) \cdot \EE[\opt(M,16B')].\]
    Now observe that $M$ is a matching and every edge has size in $(B/2, B]$, so in particular $M$ has max degree $1$ and the ratio of maximum to minimum edge size is $O(1)$. Thus, by \Cref{thm-nonadapt}, $\knapsack(M, 16B')$ outputs a non-adaptive policy that is $O(1)$-approximate. We complete the proof by reasoning about \knapsack via its LP guarantee (\Cref{prop_knapsack}) and budget-scaling property (\Cref{lem: knapsack_rescale}):
    \begin{align*}
        \EE[\opt(M,16B')] &= O(\EE[\knapsack(M,16B')]) = O(\lp(M, 16B'))\\
        &= O(\frac{B'}{B}) \cdot \lp(M, B) = O(\frac{B'}{B}) \cdot \EE[\knapsack(M,B)].
    \end{align*}
    To conclude, we obtain the final bound
    \[\EE[\opt(G,B')] = O(\frac{B'}{B}) \cdot \EE[\opt(M,16B')] = O((\frac{B'}{B})^2) \cdot \EE[\knapsack(M,B)],\]
    as required.
\end{proof}

To summarize, we gave two different algorithms for graphs which are collections of stars. If all edges have size $s_e \leq B/2$, then we designed an \ee policy, \staralg, which uses adaptive submodular maximization and the conditional knapsack LP (\ref{lp-cond}); see \Cref{thm:starsmall}. Otherwise, if all edges have size $s_e \in (B/2, B]$, then our standard non-adaptive algorithm, using the standard knapsack LP (\ref{lp_knapsack}), is already a good algorithm (\Cref{lem:starbig}). In the next section, we give our final algorithm for general graphs by decomposing the input graph into collections of stars with either small or large edges.

\subsection{Putting it all together: Algorithm for general graphs}

Finally, we are ready to prove the main result of the paper, \Cref{thm-and-adap}, which we recall here.

\thmandadap*

Our approach to proving \Cref{thm-and-adap} is natural: Using the decomposition lemma, \Cref{lem-decomp}, decompose $G$ into collections of stars that either have all small ($s_e \leq B/2$) or all large ($s_e \in (B/2, B]$) edges. Then we pick a random subcollection and run the appropriate algorithm -- either \staralg or \knapsack -- based on whether all edges are small or large, respectively. By averaging, our overall approximation ratio will be determined by how many collection of stars are needed to cover $G$. Using the Nash-Williams tree packing theorem, we can show that $O(dens(G))$ collections suffice \cite{nash1964decomposition}.

\begin{lem}\label{lem-decomp-dens}
	Given any graph $G$, there exists a poly-time algorithm that outputs $O(dens(G))$-many collections of stars each on vertex set $V(G)$ such that these collections cover all edges $E(G)$.
\end{lem}

Before proving the lemma, we use it to give our complete algorithm.

\medskip
\hrule
\smallskip
\hrule
\medskip
\genalg: Given input graph $G$ and budget $B \geq 0$:
\begin{itemize}
	\item Run the algorithm guaranteed by \Cref{lem-decomp-dens} on $G$ to obtain $O(dens(G))$-many collections of stars, say $G_1, \dots, G_\ell$ with $\ell = O(dens(G))$ such that $E(G_1) \cup \cdots \cup E(G_\ell) = E(G)$.
    \item For each collection of stars, say $G_k$, let $G_k^S$ and $G_k^L$ be the subgraphs of $G_k$ consisting of only the small ($s_e \leq B/2$) and large ($s_e \in (B/2, B]$) edges, respectively.
	\item Pick a collection of stars from $\{G_1^S, \dots, G_\ell^S\} \cup \{G_1^L, \dots, G_\ell^L\}$ uniformly at random. Let $G'$ be the chosen subgraph.
    \item If $G'$ has small edges, then run $\staralg(G', B)$. Otherwise, $G'$ has large edges, so run $\knapsack(G',B)$.
\end{itemize}
\hrule
\smallskip
\hrule
\medskip

The proof of \Cref{thm-and-adap} follows quickly from the guarantees of \staralg and \knapsack and the decomposition lemma.

\begin{proof}[Proof of \Cref{thm-and-adap}]
	We say that $G'$ is \emph{small} if $G' = G_k^S$ for some $k$ and \emph{large} if $G' = G_k^L$ for some $k$. Then conditioned on the choice of $G'$,the expected reward of our algorithm is 
    \[\EE[\genalg(G, B) \mid G'] = \ind_{G' \text{ small}} \cdot \EE[\staralg(G', B) \mid G'] + \ind_{G' \text{ large}} \cdot \EE[\knapsack(G', B) \mid G'].\]
    Applying the guarantees of \staralg (\Cref{thm:starsmall}) and \knapsack (\Cref{lem:starbig}) for any budget $B' \geq B$, we obtain
    \begin{align*}
        \EE[\genalg(G, B) \mid G'] &= \Omega((\frac{B}{B'})^2) \cdot(\ind_{G' \text{ small}} \cdot  \EE[\opt(G', B) \mid G'] + \ind_{G' \text{ large}} \cdot \EE[\opt(G', B) \mid G'])\\
        &=  \Omega((\frac{B}{B'})^2) \cdot \EE[\opt(G', B) \mid G'].
    \end{align*}
    
    To complete the proof, we apply the decomposition lemma (\Cref{lem-decomp}) and average over the $O(dens(G))$-many possible choices of $G'$. Taking $B' = 2S \cdot B$ gives
    \begin{align*}
        \EE[\genalg(G,B)] &= \EE[ \EE[\genalg(G', B) \mid G']]\\
        &= \Omega(\frac{1}{S^2}) \cdot \EE[\EE[\opt(G', 2S \cdot B) \mid G']]\\
        &= \Omega(\frac{1}{S^2 \cdot dens(G)}) \cdot (\sum_{k = 1}^\ell \EE[\opt(G_k^S, 2S \cdot B) + \EE[\opt(G_k^L, 2S \cdot B)])\\
        &= \Omega(\frac{1}{S^2 \cdot dens(G)}) \cdot \EE[\opt(G,B)]
    \end{align*}
	as required.
\end{proof}

Now we go back and prove \Cref{lem-decomp-dens}. Our algorithm first decomposes the input graph into forests, and then further decomposes each forest into collections of stars. Such a decomposition can be found, e.g. by a matroid partitioning algorithm \cite{edmonds1965minimum}.

\begin{prop}\label{prop: decomp}    
    Given graph $G$, there is a poly-time algorithm that partitions $E(G)$ into $O(dens(G))$-many forests.
\end{prop}

Using this subroutine, we state the algorithm guaranteed by \Cref{lem-decomp-dens}, \decompalg and prove that is has the desired properties.

\medskip
\hrule
\smallskip
\hrule
\medskip
\decompalg: Given input graph $G$:
\begin{itemize}
	\item Run the algorithm guaranted by \Cref{prop: decomp} on $G$ to partition the $E(G)$ into $\ell = O(dens(G))$-many forests, say $F_1, \dots, F_\ell$, each on vertex set $V = V(G)$.
	\item For each such forest, say $F_i$:
	 
	 \begin{itemize}
	 	\item Root each tree of $F_i$ arbitrarily.
	 	\item For every vertex $v \in V$, we define its depth, $d_i(v)$, to be its distance to the root of its tree in forest $F_i$.
	 	\item Then, define the three sets of vertices, $V_i(a) = \{v \in V \mid d_i(v) = a \mod 3\}$ for $a = 0, 1, 2$.
	 	\item Define the three vertex-induced subgraphs $F_i(a) = F_i - V_i(a)$ for $a = 0,1, 2$.  See \Cref{fig:tree_ex} for example.
 	\end{itemize}
 	\item Output the collection of subgraphs $\{F_i(a) \mid i = 1, \dots, \ell ;\, a = 0,1,2\}$.
\end{itemize}
\hrule
\smallskip
\hrule
\medskip

    \begin{figure}[h]
        \centering
        \includegraphics[scale = 0.3]{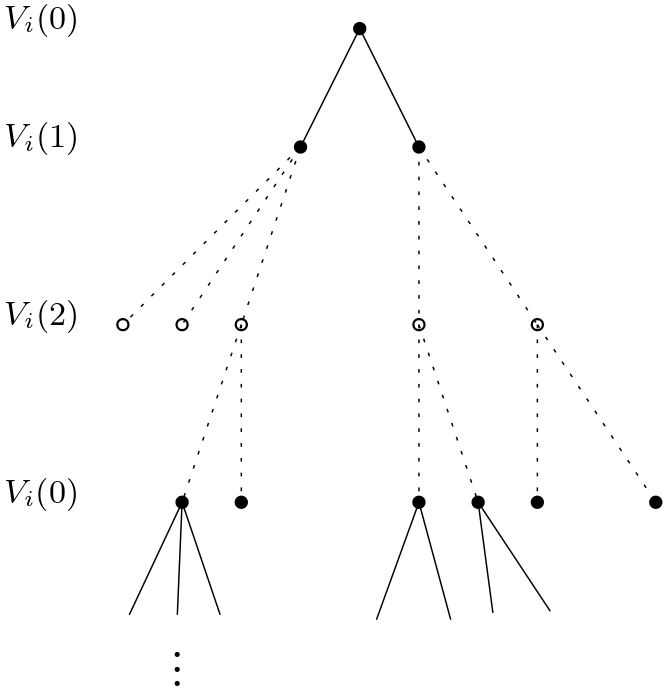}
        \caption{An example of a single tree in subgraph $F_i(2)$. Solid edges and vertices appear in $F_i(2)$.}
        \label{fig:tree_ex}
    \end{figure}

\begin{proof}[Proof of \Cref{lem-decomp-dens}]
	It is clear that \decompalg is efficient, and we output $O(dens(G))$-many subgraphs, so it remains to show that the collection $\{F_i(a) \mid i = 1, \dots, \ell ;\, a = 0,1,2\}$ covers all edges of $G$. Consider any edge $e$ of $G$. Note that the collection $\{F_i \mid i = 1, \dots, \ell\}$ covers $G$, so $e$ belongs to some forest, say $F_i$. Since $F_i$ is a forest, the edge $e$ must connect vertices in two consecutive depth levels. Modulo $3$, this leaves one choice of $a= 0,1,2$ such that $e$ connects two vertex whose depth is \emph{not} equal to $a \mod 3$. We conclude, $e$ is an edge in $F_i(a)$ for this choice of $a$.
\end{proof}

    \section{Conclusion and Future Work}

We introduced \bayes, one of the first stochastic combinatorial optimization models that incorporates correlated rewards, while still admitting provably efficient approximation algorithms. Our non-adaptive and adaptive algorithms give the first poly-time guarantees for Bayesian active search (and stochastic probing and knapsack with correlations across items). These results show that realistic dependencies among items need not preclude strong worst-case performance. Further, our correlation model is rich enough to capture typical bad examples for Bayesian active search as well as practically-relevant prior distributions, such as those arising in site percolation.

Our work gives the foundation for how techniques from stochastic combinatorial optimization can have a large impact in Bayesian machine learning problems, and it opens up many directions for future work:
\begin{itemize}
    \item \textbf{Hypergraphs} Extending our model from graphs to hypergraphs would allow us to capture higher-order correlations common in biological and social networks.
    
    \item \textbf{Different reward functions} In \bayes, one can view the reward of an edge as the (weighted) \textsc{AND} of its endpoints, but it would be interesting and practically relevant to consider a generalization where the reward of a (hyper-)edge is some Boolean formula applied to its endpoints. For example, threshold functions (at least $t$ endpoints are active) would be an interesting generalization.

    \item \textbf{Different probes} Further, in our model, when we probe an edge we learn the outcomes of both of its endpoints. One could consider a more restricted setting where an edge probe only reveals the realized edge reward (in \bayes, this means one would only learn the \textsc{AND} of the endpoints, rather than both endpoint values). This would introduce much more complicated Bayesian updates.
\end{itemize}

We hope the techniques developed here -- particularly the decomposition lemma and our structural properties of \ee policies -- serve as building blocks for tackling richer correlated stochastic optimization problems.

\section*{Acknowledgments}

We thank Viswanath Nagarajan for helpful discussions leading to the connection with adaptive submodularity, and an anonymous reviewer for pointing out the alternative adaptivity gap lower bound in \Cref{lem_matching_lower}.  AG's research was supported in part by NSF awards CCF-2422926 and CCF-2608359. Benjamin Moseley was supported in part by ONR grant N000142212702.




    \bibliographystyle{plain}
    \bibliography{ref}

    \appendix

    \section{Basic properties of knapsack LP}\label{sec: appendix_knapsack}

Here we prove some basic facts about the knapsack LP (\ref{lp_knapsack}) and (\ref{lp-cond}), which we need throughout our analysis.

\lemknapsackrescale*
\begin{proof}
    Consider an optimal solution, say $x^*$, to $\lp(G,B')$. We claim that $\frac{B}{B'} \cdot x^*$ is feasible for $\lp(G,B)$. For the knapsack constraint, we have
    \[\sum_{e \in E(G)} s_e \cdot \frac{B}{B'} x^*_e =  \frac{B}{B'} \sum_{e \in E(G)} s_e \cdot x^*_e \leq B.\]
    Using our assumption that all edges have size at most $B$, the pruning constraint ($x_e = 0$ if $s_e > B$) of $\lp(G,B)$ is trivial, and it is clear that $0 \leq \frac{B}{B'} \cdot x^*_e \leq 1$ for all edges $e$. This solution to $\lp(G,B)$ achieves objective value $\frac{B}{B'} \cdot \lp(G,B')$, so we have $\lp(G,B) \geq \frac{B}{B'} \cdot \lp(G,B')$, as required.
\end{proof}

\propsubmod*
\begin{proof}
    Consider any sets $C_1\subset C_2 \subset C$ and adding any center $c \notin C_2$ to sets $C_1$ and $C_2$. We will show that $g(C_1 \cup \{c\}) - f(C_1) \geq f(C_2 \cup \{c\}) - f(C_2)$.
    
    Here, we use the well-known fact that an optimal solution to the knapsack LP is obtained by the \emph{fractional greedy algorithm}, $\fracgreedy(P)$ \cite{dantzig1957discrete}. Stated in terms of $\lp(G \mid P,B)$, we consider edges in $E(\bigcup_{c \in P} G_c)$ in decreasing order of value density $\frac{\EE[R_e]}{s_e}$. Initially all $x_e$'s are set to $0$. For each considered edge $e$, we increase $x_e$ as much as possible as long as the budget constraint is satisfied subject to the LP constraint $x_e \leq 1$. In the end, this leads to a solution with at most one fractional $x_e$.
    
    We first interpret the quantity $f(C_2 \cup \{c\}) - f(C_2) = \fracgreedy(C_2 \cup \{c\}) - \fracgreedy(C_2)$ in terms of the fractional greedy algorithm. Consider $\fracgreedy(C_2 \cup \{c\})$. In this algorithm, suppose we allocate $b$ units of budget to edges in $G_c$ (so $b = \sum_{e \in E(G_c)} s_e \cdot x_e$). By definition of \fracgreedy, these $b$ units are allocated to the most value-dense edges in $G_c$ (those with largest $\frac{\EE[R_e]}{s_e}$). Let $u$ be the reward obtained by choosing these edges. These edges replace the $b$ units of the \emph{least} value-dense edges allocated by $\fracgreedy(C_2)$, which give reward, say $\ell_2$. Thus, we have $f(C_2 \cup \{c\}) - f(C_2) = u - \ell_2$.

    Next, we lower bound $f(C_1 \cup \{c\}) - f(C_1)$ similarly. Consider the fractional solution generated by $\fracgreedy(C_1)$. We lower bound $f(C_1 \cup \{c\})$ as follows: Replace the $b$ units of the least value-dense edges allocated by $\fracgreedy(C_1)$, which give reward, say $\ell_1$ with the $b$ units of the the most value dense edges in $G_c$, which again gives reward $u$. This gives a feasible fractional solution to $\lp(G \mid C_1 \cup \{c\}, B)$ with objective value $f(C_1) + u - \ell_1$. Since $C_1 \subset C_2$, we have $\ell_1 \leq \ell_2$, so we can conclude $f(C_1 \cup \{c\}) - f(C_1) \geq u - \ell_1 \geq u - \ell_2 = f(C_2 \cup \{c\}) - f(C_2)$, as required.
\end{proof}

\knapsackgap*
    \begin{proof}
        The assumed instance for \adaptiveknapsack is exactly an instance of the following stochastic probing problem, \probe by taking the function $g(P) = \lp(G \mid P, B)$ and the set family of all subsets of $C$ with total size at most $B'$.
        \begin{definition}[\probe]\label{def-probe}
        	In an instance of \probe, we are given a set function $f:2^Y \rightarrow \RR_+$ on ground set $Y$ and a downward-closed set family $\mathcal{Y} \subset 2^Y$. Every item $y \in y$ is independently \emph{active} with known probability $p_y$. However, whether an item is active or not is unknown until that item is \emph{probed}. The goal is to adaptively probe a set of elements $P \in \mathcal{Y}$ and then choose a subset of the active probed elements, say $Q \subset P \cap A$, where $A \subset Y$ is the subset of active elements, to maximize $\EE[f(Q)]$. 
        \end{definition}
        \probe has the following known adaptivity gap result due to \cite{DBLP:conf/soda/GuptaNS16}.
        \begin{prop}\label{prop-probe-gap}
        	If the set function $f$ is non-negative, monotone, and submodular, then the adaptivity gap of \probe is at most $3$.
        \end{prop}
        It is clear that our function $g$ is monotone (since adding more items to the knapsack instance can only improve the LP objective value), and it is submodular by \Cref{prop:submod}. The only difference between the assumed instance of \adaptiveknapsack and \probe is that in the latter, one can choose a subset $Q \subset P \cap A$ to receive the reward from. However, since $g$ is monotone, we can assume that one always chooses $Q = P \cap A$. Thus, \Cref{prop-probe-gap} gives the desired adaptivity gap.
    \end{proof}

    \section{\bayes is not adaptive submodular}\label{appendix: adapt}

In this section, we show that the prior distribution and objective function of \bayes do not satisfy adaptive submodularity. As in \Cref{sec-adapt}, we do not define adaptive submodularity in full generality. Rather, we consider a specialized definition for our purposes.

Given instance of \bayes on graph $G = (V,E)$, the reward of probing edges $P \subset E$ given that the active vertices are $A \subset V$ is given by the function $f(P, A) = \sum_{uv \in P} w_{uv} \cdot \ind_{u \in A} \ind_{v \in A}$. We consider the prior distribution where each vertex $v \in V$ is in $A$ independently with probability $p_v$.

In this setting, a partial realizations of edge set $E' \subset E$ is subset of active vertices $A_{V(E')} \subset V$, where $V(E') \subset V$ is all vertices incident to some edge in $E'$. We say two partial realizations $A_{V(E_1)}$ and $A_{V(E_2)}$ of $E_1 \subset E_2 \subset E$ are consistent if $A_{V(E_2)} \cap V(E_1) = A_{V(E_1)}$.

Then we define, the marginal gain of adding edge $e \in E$ to a subset $E' \subset E$ of already-probed edges with partial realization $A_{V(E')}$ by
\[\gain(e, E' \mid A_{V(E')}) = \EE_{A}[f(E' \cup \{e\}, A) \mid A \cap V(E') = A_{V(E')}] - \EE_{A}[f(E', A) \mid A \cap V(E') = A_{V(E')}].\]
Then the function $f$ is adaptive submodular if for any subsets $E_1 \subset E_2 \subset E$ and consistent partial realizations $A_{V(E_1)}$ and $A_{V(E_2)}$  of $E_1$ and $E_2$, respectively, we have
\[\gain(e, E_1 \mid A_{V(E_1)}) \geq \gain(e, E_2 \mid A_{V(E_2)}).\]
We will show that \bayes is far from satisfying this condition in general.
\begin{lem}\label{lem:bayeslower}
    For any $\epsilon \in (0,1)$, there exists an instance of \bayes along with an subsets $E_1 \subset E_2 \subset E$, $e \notin E_2$, and consistent partial realizations $A_{V(E_1)}$ and $A_{V(E_2)}$  of $E_1$ and $E_2$, respectively such that
    \[\gain(e, E_1 \mid A_{V(E_1)}) \leq \epsilon \cdot \gain(e, E_2 \mid A_{V(E_2)}).\]
\end{lem}
\begin{proof}
    Consider a star with two leaves such that the center $c$ has probability $\epsilon$ both leaves, say $a$ and $b$ have probability $1$, and both edges (one to each leaf) has weight $1$. Then take $E_1 = \emptyset$, $E_2 = ca$,  and $e = cb$. For the partial realizations, $A_{V(E_1)} = \emptyset$ is trivial, and we take $A_{V(E_2)} = \{c,a\}$ (i.e. the center $c$ is active). Then we have
    \[\gain(e, E_1 \mid A_{V(E_1)}) = \EE_{A}[\ind_{c \in A} \ind_{b \in A}] = \Pr(c \in A) = \epsilon.\]
    On the other hand, the marginal gain of probing $cb$ when we already know that the center is active is
    \[\gain(e, E_2 \mid A_{V(E_2)}) = \EE_{A}[\ind_{c \in A} \ind_{b \in A} \mid A \cap \{c,a\} = \{c,a\}] = 1,\]
    as required.
\end{proof}

    \section{Alternate Adaptivity Gap Lower Bound}\label{sec_matching_lower}

\lemmatchinglower*
\begin{proof}
    We define an instance of \bayes with $G = (V,E)$, where $E = K \cup M$, where $K$ is a complete graph on vertex set $V$ and $M$ is a perfect matching on $V$ (assume $\lvert V \rvert$ is even). We write $n = \lvert V \rvert$. Then for each edge $e \in K$, we set $s_e = n$ and $w_e = 1$. For each edge $e \in M$, we set $s_e = \frac{2}{n}$ and $w_e = 0$. For every vertex $v \in V$, we set $p_v = \frac{1}{n}$. We take the query budget to be $B = n + 1$. This completes the description of our instance.

    Note that $\lvert E \rvert = \Theta(n^2)$, so it suffices to show an adaptivity gap lower bound of $\Omega(n^2)$. First, we analyze the optimal non-adaptive policy. We observe that every edge in $K$ gives expected reward $\frac{1}{n^2}$, and every edge in $M$ gives expected reward $0$. Thus, $\opt_{NA}$ queries a single edge from $K$ to obtain expected reward $\EE[\opt_{NA}] = \frac{1}{n^2}$.

    On the other hand, consider the following adaptive policy: First query all $\frac{n}{2}$-many matching edges, which costs $\frac{n}{2} \cdot {2}{n} = 1$ unit of the budget. Afterwards, we know exactly which vertices are active/inactive. If there are at least $2$ active vertices, which happens with probability $\Omega(1)$, then we spend the remaining $n$ units of budget to query a single such edge from $K$ to obtain reward $1$. We conclude, $\EE[\opt_A] = \Omega(1)$, which gives the desired gap.
\end{proof}

\end{document}